\documentclass[11pt]{article}

\usepackage{amssymb}
\usepackage{graphicx}
\usepackage{amsmath}
\usepackage{makeidx}
\usepackage{indentfirst}
\usepackage[T1]{fontenc}
\usepackage[utf8]{inputenc}
\usepackage{authblk}

\newcounter{resultnum}[section]
\newtheorem{conclusion}{Conclusion}[section]

\newcounter{conclusionnum}[section]
\newcounter{conditionnum}[section]
\newcounter{conjecturenum}[section]
\newcounter{examplenum}[section]
\newcounter{exercisenum}[section]
\newtheorem{lemma}{Lemma}[section]

\newcounter{lemmanum}[section]
\newcounter{notationnum}[section]
\newtheorem{theorem}{Theorem}[section]

\newcounter{theoremnum}[section]
\newtheorem{definition}{Definition}[section]

\newcounter{definitionnum}[section]
\newtheorem{corollary}{Corollary}[section]

\newcounter{corollarynum}[section]
\newtheorem{remark}{Remark}[section]

\newcounter{remarknum}[section]
\newtheorem{proposition}{Proposition}[section]

\newcounter{propositionnum}[section]
\newcounter{acknowledgementnum}[section]
\newcounter{algorithmnum}[section]
\newcounter{axiomnum}[section]
\newcounter{casenum}[section]
\newtheorem{claim}{Claim}[section]

\newcounter{claimnum}[section]
\newcounter{summarynum}[section]
\newcounter{problemnum}[section]
\newenvironment{proof}[1][]{\textbf{Proof.} }{}

\begin{document}

\title{Spinor and Twistor Geometry in\\ Einstein Gravity and Finsler Modifications}
\date{June 21, 2014}

\author{\textbf{Sergiu I. Vacaru} \thanks{%
sergiu.vacaru@uaic.ro;\
http://www.uaic.ro/uaic/bin/view/Research/AdvancedTheoretical}}

\affil{\small Theory Division, CERN, CH-1211, Geneva 23, Switzerland\footnote{associated visiting research affiliation};\ and
\newline
\small Rector's Office, Alexandru Ioan Cuza University,  Alexandru
Lapu\c sneanu \newline street, nr. 14,  UAIC -- Corpus R, office 323;  Ia\c
si,\ Romania, 700057 }

\renewcommand\Authands{and }

\maketitle

\begin{abstract}
We present a generalization of the spinor and twistor geometry for  on (pseudo) Riemannian manifolds enabled with nonholonomic distributions or for Finsler--Cartan spaces modelled on tangent Lorentz bundles. Nonholonomic (Finsler) twistors are defined as solutions of generalized twistor equations determined by spin connections and frames adapted to nonlinear connection structures. We show that the constructions for local twistors can be globalized using nonholonomic deformations with  "auxiliary" metric compatible connections completely determined by the metric structure and/or the  Finsler fundamental function. We explain how to perform such  an approach in the Einstein gravity theory formulated in Finsler like variables with conventional nonholonomic 2+2 splitting.

\vskip0.1cm


\vskip3pt


\end{abstract}



\section{Introduction}

Twistor theory began with R. Penrose's two papers in 1967 and 1968 and the
subject has grown in different directions of modern mathematics and
classical and quantum physics (for introductions to twistor theory and
reviews of results, and references, see the monographs \cite%
{penrr,manin,ward,huggett}). From a "modest" geometric point of view,
twistor structures and transforms are naturally related to certain methods
of constructing solutions for self--dual Yang--Mills and Einstein equations
and complex/supersymmetric generalizations of the Minkowski and Einstein
spacetime geometry.

In this article, we give an introduction into the differential geometry of
Finsler spinors and twistors. We shall define twistors for models of metric
compatible Finsler spaces and study possible connections to the general
relativity theory and modifications. Our "pragmatic" goal is to point
researchers that there is an important relation between the twistor
transforms, nonholonomic deformations of fundamental spacetime geometric
objects and a method for generating exact off--diagonal solutions of
gravitational field equations in Einstein gravity and modified/generalized
theories.\footnote{%
The corresponding metrics can not be diagonalized via coordinate transforms.
In modern literature, there are used different words
"anholonomic", "nonholonomic" and/or "non--integrable" which we shall
consider as equivalent ones.} We shall apply a geometric formalism related
to the anholonomic frame deformation method of constructing exact solutions
\cite{veyms} and A--brane, deformation and gauge like
quantization of gravity \cite{vgauge}. In such models, Finsler
like variables can be introduced via respective nonholonomic 2+2 and/or $n+n$
splitting, for distributions with fibered structure, on Einstein and/or
(pseudo) Riemannian manifolds. This allows us to decouple the Einstein
equations with respect to some classes of nonholonomic frames and construct
generic off--diagonal exact solutions depending on all coordinates via
certain classes of generating and integration functions and parameters.

Let us emphasize some substantial differences between the geometry of
(pseudo) Finsler spaces and that of (pseudo) Riemannian manifolds\footnote{%
In our works, the meaning of the word "pseudo" is equivalent to "semi" which
are used for models of curved spaces in standard particle and gravity
physics and/or in mathematical literature when the metric may have a local
pseudo--Euclidean local signature of type $(-,+,+,+)$. Here we note that
such terms, together with a nonstandard for physicists concept of Minkowski
space, have different definitions in some monographs on Finsler geometry
\cite{rund,matsumoto,bejancu1,bejancu2,bcs}, see also recent theories
with (pseudo) Finsler metrics, for instance, \cite%
{vsgg,stavrinos,stavrinosk,mavromatos,visser}.}, see definitions and
details in next section. Different models of Finsler geometry are
characterized by different classes of nonlinear and linear connections and
lifts on tangent bundles of geometric objects,   nonholonomic distributions
and related curvature and torsion tensors. Such a geometry is not completely
determined by a Finsler metric $F(x^{i},y^{a})$, which is a nonlinear
quadratic element with homogeneity conditions of typical fiber coordinates $%
y^{a}$ on a tangent bundle $TM$ to a manifold $M,$ with local coordinates $%
x^{i},$ or a (nonholonomic) manifold $\mathbf{V}$ with non--integrable
fibred structure. The geometric constructions on Finsler spaces have to be
adapted to an another fundamental geometric object, the nonlinear connection
(N--connection), $\mathbf{N.}$ There are necessary additional assumptions on
a chosen Finsler linear connection $\mathbf{D}$ (the third fundamental
geometric object, also adapted to the N--connection) which can be metric
compatible, or non--compatible.\footnote{%
We use boldface symbols for spaces enabled with a N--connection structure
and geometric objects adapted to such a structure \cite{vrflg}. In
certain canonical models of Finsler geometry, the values $\ _{F}\mathbf{N},
\ _{F}\mathbf{g},\ _{F}\mathbf{D}$ are derived for a fundamental
(generating) Finsler function $F$ following certain geometric/variational
principles, up to certain classes of generalized frame transforms. For
simplicity, in this work we shall omit left up or low labels if that will
not result in ambiguities. We shall work on necessary type real and/or
complex manifolds of finite dimensions. Tensor and spinor indices will be
considered, in general, as abstract ones, and with the Einstein's summation
rule for frame/coordinate reprezentations.} In certain cases, it is possible
to construct an associated metric structure $\mathbf{g}$ on total space $TM$%
, to introduce various types of curvature, torsion and nonmetricity tensors
etc.

In brief, a model of Finsler geometry is completely stated by a triple
of fundamental "boldface" geometric objects $\left( F:\mathbf{N,g,D}\right)$
generated by $F $, which is very different from the case of (pseudo)
Riemannian geometry which is determined by data $(\mathbf{g,}%
\bigtriangledown ),$ for a metric tensor $\mathbf{g}$ being compatible with
the Levi--Civita connection $\bigtriangledown $ and with zero torsion (such
a linear connection is completely defined by the metric structure).

Some classes of Finsler geometries and generalizations may admit spinor
formulations \cite{vspinor,vhep,vstavr}. For instance, in the case of metric
compatible models, such methods were developed for the so--called
Finsler--Cartan and canonical distinguished connections \cite%
{cartf,vsgg,vrflg}) when well defined Finsler--Ricci flow \cite{vfricci},
supersymmetric \cite{vsuperstr} and/or noncommutative generalizations \cite%
{vfncr} can be performed and applied for generating exact solutions for
Finsler brane cosmology \cite{vfbr}  etc. Such constructions are
technically very cumbersome and, in general, not possible for metric
noncompatible Finsler geometries with, for instance, Chern or Berwald
connections, see \cite{vrflg,vcrit} and references therein for reviews
of results and critical remarks on possible applications in modern physics.

A proposal how the concept of twistors and nonholonomic twistor equations
could be extended for Finsler spaces and generalizations was discussed in
\cite{vhep}. The approach was based on an idea to use spinor and twistor
geometries for arbitrary finite dimensional spacetimes (see Appendix to
volume 2 in \cite{penrr}) modified to the case of metric compatible Finsler
connections and their higher order generalizations with "nonhomogeneous"
generating functions, for instance, generalized Lagrange geometries etc.
Those nonholonomic (generalized Finsler) twistor constructions had some
roots to former our works on "twistor wave function of Universe" \cite{vcluj}
 and nearly
geodesic maps of curved spaces and twistors \cite{vostaf}.\footnote{\textbf{%
Historical Remarks:}\ The author of this paper began his research on
"twistor gauge models of gravity" when he was a post--graduate at physics
department of "M. A. Lomonosov" State University at Moscow, during
1984-1987. At that time, there were translated in Russian some fundamental
papers on Twistor Theory and published a series of important works by
"soviet" authors, for instance, the Russian variant of \cite{manin}. Various
subjective issues related to the "crash" of former Soviet Union resulted in
defending author's PhD thesis \cite{vphd} in 1994, at University Alexandru
Ioan Cuza, UAIC, at Ia\c si (Yassy), Romania. There were also certain
important scientific arguments to transfer such a research on geometry and
physics from Russia to Romania, where some schools on nonholonomic geometry
and generalized Finsler spaces (for instance, supervised by G. Vr\v anceanu and  A. Bejancu) published a number of works beginning 20ths of previous Century but
worked in isolation during dictatorial "socialist period".
\par
Twistor equations are generically nonintegrable for arbitrary curved
spacetime. Such equations became integrable, for instance, if the Weyl
spinor vanishes, see details in \cite{penrr}. Our main idea is to use for
definition of twistors another class of equations with an "auxiliary" metric
compatible connection, completely determined by the same metric but with
nontrivial torsion. For certain conditions on noholonomic structure, various
"non--integrable" twistor configurations can be globalized in a
self--consistent form. Imposing additional anholonomi conditions, we can
extract, for instance, certain real vacuum Einstein manifolds. Such methods
were formalized by G. Vr\v anceanu in his geometry of "nonholonomic
manifods" \cite{vranceanu1,vranceanu2,vranceanu3}, see further developments
in \cite{bejancu3} and, with applications in modern classical and quantum
gravity and generalized Finsler geometries and (non) commutative geomeric
flows, \cite{vsgg,vrflg,vfncr}.
\par
In some sense, Finsler spaces are modelled by geometries with nonholonomic
distributions on tangent bundles or on manifolds with fibered structure. The
most closed to standard physics directions
were developed following some fundamental geometric ideas and results due to
E. Cartan \cite{cartf}, M. Matsumoto \cite{matsumoto} and others (on almost K%
\"{a}hler Finsler structures, Einstein equations for the Cartan--Finsler
connection, J. Kern's geometrization of mechanics via Finsler methods \cite%
{kern} etc). Using nonholonomic distributions, it was possible to formulate
Clifford and spinor analogs of metric compatible Finsler geometries and
generalizations and define Dirac--Finsler operators \cite{vspinor,vhep}, see
also reviews of results and complete lists of references in \cite%
{vstavr,vsgg,vfncr}. Approaches based on Berwald and Chern connections,
semisprays and sectional curvature \cite{bcs,akbar} have been developed in
modern literature but because of nonmetricity and spectific "nonminimal"
relations between the Finsler metric and curvature seem to be less related
to standard theories in physics \cite{vcrit,vrflg}.}

Recently, it was suggested \cite{dunajski} to use the nondegenerate Hessian $%
g_{ab}:=\frac{1}{2}\frac{\partial ^{2}F}{\partial y^{a}\partial y^{b}}$ in
order to define a class of twistor structures related to Finslerian
geode\-sics, which are integral curves of certain systems of ordinary
differential equations (ODEs) and some projective classes of isotropic
sprays. Such constructions were performed for Finsler generating functions
with scalar flag curvature, in special, for the Randers metric. It was
possible to formulate a variational principle for twistor curves arising
from such examples of Finsler geometries with scalar flag structure. The
(nonlinear) geodesic/ (semi) spray configurations reflect only partially the
geometric properties, and possible relations to systems of differential
equations (with partial derivatives ones, PDEs, and/or ODEs) of spaces
endowed with fundamental Finsler functions. As we emphasized above, complete
geometric and possible physically viable models of Finsler spacetimes can be
formulated after additional assumptions on linear connection structures and
conditions on their (non) compatibility with the metric and N--connection
structure. Only in such cases, we can conclude if there are, or not, Finsler
analogs of spinors and anisotropic models of bosonic and fermionic fields
and interactions for certain classes of Finsler connections.

This paper is organized as follows: In section \ref{nfgg}, we present a
brief introduction into the geometry of metric compatible nonholonomic
manifolds and bundles enabled with nonlinear connection (N--connection)
structure. It is elaborated an unified N--adapted formalism both
for the Finsler--Cartan spaces and the Einstein gravity theory reformulated
in Finsler like variables. We also provide new results on the geometry of
conformal transforms and N--connection structures. Section \ref{snfs} is
devoted to the  differential geometry of spinors on Finsler--Cartan and Einstein--Finsler spaces. We define nonholonomic (Finsler) twistors in section \ref{snft} considering "auxiliar" metric compatible Finsler like connections. There are studied the conditions when nonholonomic twistors can describe global structures in Einstein gravity and modifications.

\vskip5pt \textbf{Acknowledgements: } The work is partially supported by the Program
IDEI, PN-II-ID-PCE-2011-3-0256 and by an associated visiting research
position at CERN.

\section{Nonholonomic (Finsler) Geometry and Gravity}

\label{nfgg}

Finsler type geometries can be modelled on (pseudo) Riemannian manifolds
and/or tangent bundles enabled with necessary types nonholonomic
distributions and supersymmteric and/or noncommutative generalizations. In
this section, we fix notations and provide necessary results. For details
and proofs, we refer to \cite{vdefq,vrflg,vsgg} where the so--called
geometry of Finsler--Einstein gravity and modifications is formulated in a
language familiar to researchers in mathematical relativity.

\subsection{The geometry of nonholonomic bundles and manifolds}

\subsubsection{Nonlinear and distinguished connections}

Let us consider a $(n+m)$--dimensional real or complex manifold $\mathbf{V}$
of necessary smooth/analytic/holomorphic class, where 1) $n=m=2$ (this will
be used for constructing exact solutions in general relativity applying
Finsler and twistor methods), or 2) $n=m=4$ (for Finsler twistor models on
tangent bundles to Lorentz manifolds).

\begin{definition}
A nonholonomic manifold is a pair $(\mathbf{V},\mathcal{N}) $ defined by a
nonintegrable distribution $\mathcal{N}$ on $\mathbf{V.}$
\end{definition}

The geometry of nonholonomic real manifolds is studied in Refs. \cite%
{vranceanu1,vranceanu2,vranceanu3,vrflg,vsgg} for various finite dimensions $%
n\geq 2$ and $m>1.$ For simplicity, we shall consider a subclass of
nonholonomic distributions $\mathcal{N}$ stating fibered structures $\pi :$ $%
\mathbf{V}\rightarrow V$ with constant rank $\pi ,$ where $V$ is a two
dimensional, 2-d, or 4-d, (for instance, pseudo--Riemanian manifold, or any
its complexified version). In general, we can consider any such map with
differential map $\pi ^{\intercal }:$ $T\mathbf{V}\rightarrow TV$ when the
kernel of $\pi ^{\intercal },$ which is just the vertical subspace $v%
\mathbf{V}$ with a related inclusion mapping $i:\ v\mathbf{V\rightarrow }T%
\mathbf{V,}$ defines a corresponding vertical subspace as a nonholonomic
distribution.

\begin{definition}
\textbf{--Theorem:}\label{deftnc} A nonlinear connection (N--connection) $%
\mathbf{N}$ on $\mathbf{V}$ can be defined in two equivalent forms:\newline
a) by the splitting on the left with an exact sequence $0\rightarrow v%
\mathbf{V}\overset{i}{\rightarrow }T\mathbf{V\rightarrow } $ $T\mathbf{V/}v%
\mathbf{V}\rightarrow 0,$ i.e. by a morphism of submanifolds $\mathbf{N}:\ T%
\mathbf{V\rightarrow }v\mathbf{V}$ such that $\mathbf{N\circ i}$ is the
unity in $v\mathbf{V.}$\newline
b) Globalizing the local distributions associated to such nonholonomic
splitting $\mathcal{N}$ we prove that a N--connection defines a Whitney sum
of conventional horizontal (h) subspace, $h\mathbf{V,}$ and vertical (v)
subspace, $h\mathbf{V,}$%
\begin{equation}
T\mathbf{V=}h\mathbf{V\oplus }v\mathbf{V.}  \label{whitney}
\end{equation}
\end{definition}

We shall use "boldface" symbols in order to emphasize that a geometric
object is defined on a nonholonomic manifold $\mathbf{V}$ enabled with
N--connection structure and call such an object to be, for instance, a
distinguished tensor (in brief, d--tensor, d-metric, d--spinor,
d--connection), or a d--vector $\mathbf{X}=(h\mathbf{X},v\mathbf{X})\in T%
\mathbf{V.}$ A N--connection is characterized by its curvature, i.e. the
Neijenhuis tensor,
\begin{equation}
\Omega (\mathbf{X,Y}):=\left[ v\mathbf{X},v\mathbf{Y}\right] +v\left[
\mathbf{X},\mathbf{Y}\right] -v\left[ v\mathbf{X},\mathbf{Y}\right] -v\left[
\mathbf{X},v\mathbf{Y}\right] ,  \label{ncurvat}
\end{equation}%
for any d--vectors $\mathbf{X,Y}$ and commutator $\left[ \cdot ,\mathbf{%
\cdot }\right] .$\footnote{%
N--connections were used in coefficient form in E. Cartan's first monograph
on Finsler geometry \cite{cartf}. The first global definition is due to C.
Ehresmann \cite{ehresmann}, it was studied in Finsler geometry and
generalizations by A. Kawaguchi \cite{kawaguchi} and Greek and Romanian geometers and physicists \cite%
{vsgg}. The abroach was developed and applied to various nonholonomic and/or
Finsler generalizations on superspaces, in noncommutative geometry and
constructing exact solutions in gravity \cite%
{vspinor,vsuperstr,vhep,vsgg,vrflg}, see also references therein.}

\begin{definition}
A distinguished connection (d--connection) $\mathbf{D}$ is a linear
connection conserving under parallelism the Whitney sum (\ref{whitney}),
i.e. the N--connection splitting into $h$-- and $v$--subspaces.
\end{definition}

We can perform a decomposition of $\mathbf{D}$ into h- and v--covariant
derivatives $\mathbf{D=}\left( h\mathbf{D},v\mathbf{D}\right) ,$ when $%
\mathbf{D}_{\mathbf{X}}:=\mathbf{X}\rfloor \mathbf{D=}h\mathbf{X}\rfloor
\mathbf{D}+v\mathbf{X}\rfloor \mathbf{D=}h\mathbf{D}_{\mathbf{X}}+v\mathbf{D}%
_{\mathbf{X}}$, where "$\rfloor "$ is the interior product.

\begin{definition}
For a d--connection $\mathbf{D,}$ we can define:\newline
a) the torsion d--tensor%
\begin{equation}
\mathcal{T}(\mathbf{X},\mathbf{Y}):=\mathbf{D}_{\mathbf{X}}\mathbf{Y-D}_{%
\mathbf{Y}}\mathbf{X-}\left[ \mathbf{X},\mathbf{Y}\right] ;  \label{tors}
\end{equation}%
b) the curvature d--tensor
\begin{equation}
\mathcal{R}(\mathbf{X},\mathbf{Y}):=\mathbf{D}_{\mathbf{X}}\mathbf{D}_{%
\mathbf{Y}}\mathbf{-D}_{\mathbf{Y}}\mathbf{D}_{\mathbf{X}}\mathbf{-D}_{\left[
\mathbf{X},\mathbf{Y}\right] }.  \label{curv}
\end{equation}
\end{definition}

Introducing $h$-$v$--decompositions $\mathbf{D}=\left( h\mathbf{D},v\mathbf{D%
}\right) $ and $\mathbf{X}=(h\mathbf{X},v\mathbf{X})$ in above formulas, we
compute respective $h$-$v$--components (i.e. d--tensor N--adapted
components) of the torsion and curvature of a d--connection. For instance,
there are five nontrivial components of torsion,{\small
\begin{equation*}
\mathcal{T}(\mathbf{X},\mathbf{Y})=\{h\mathcal{T}(h\mathbf{X},h\mathbf{Y}),h%
\mathcal{T}(h\mathbf{X},v\mathbf{Y}),v\mathcal{T}(h\mathbf{X},h\mathbf{Y}),v%
\mathcal{T}(a\mathbf{X},h\mathbf{Y}),v\mathcal{T}(v\mathbf{X},v\mathbf{Y})\},
\label{dtors1}
\end{equation*}%
} for arbitrary d--vectors $\mathbf{X}$ and $\mathbf{Y.}$

\subsubsection{Metric compatible nonholonomic manifolds}

A nonholonomic manifold $\mathbf{V}$ is enabled with a metric structure
defined by symmetric nondegenerate second rank tensor $\mathbf{g}.$ Such a
metric can be described equivalently by a d--metric $\mathbf{g}=(h\mathbf{g}%
,v\mathbf{g})$ with corresponding $h$-- and $v$--metrics, in N--adapted form.

\begin{definition}
\textbf{-Theorem:}\ A d--connection $\mathbf{D}$ is metric compatible with a
d--metric $\mathbf{g}$ if and only if $\mathbf{Dg}=0$ imposing conditions of
compatibility in $h$-$v$--form for decompositions $\mathbf{D=}\left( h%
\mathbf{D},v\mathbf{D}\right) $ and $\mathbf{g}=(h\mathbf{g},v\mathbf{g}).$
\end{definition}

In this paper, we shall use only metric compatible connections.

\begin{theorem}
Any metric structure $\mathbf{g}$ defines a unique Levi--Civita connection $%
\nabla $ which is metric compatible, $\nabla \mathbf{g}=0,$ and with zero
torsion,\newline
 $\ ^{\nabla }\mathcal{T}(\mathbf{X},\mathbf{Y}):=\nabla _{\mathbf{X}%
}\mathbf{Y-}\nabla _{\mathbf{Y}}\mathbf{X-}\left[ \mathbf{X},\mathbf{Y}%
\right] =0.$
\end{theorem}

We note that $\nabla $ is not a d--connection because it does not preserve
under parallelism the N--connection splitting. For a $\mathbf{N}$ (\ref%
{whitney}) with nonzero N--curvature $\Omega $ (\ref{ncurvat}), there is a
"preferred" d--connection which can be considered as the analog of the
Levi--Civita connection for nonholonomic manifolds:

\begin{theorem}
Any metric structure $\mathbf{g}$ defines a unique canonical d--connecti\-on
$\mathbf{D}$ which is metric compatible, $\mathbf{Dg}=0,$ and with zero
"pure" horizontal and "vertical" d--torsions, i.e., respectively, $h\mathcal{%
T}(h\mathbf{X},h\mathbf{Y})=0$ and $v\mathcal{T}(v\mathbf{X},v\mathbf{Y})$,
for $\ \mathcal{T}(\mathbf{X},\mathbf{Y}):=\mathbf{D}_{\mathbf{X}}\mathbf{Y-D%
}_{\mathbf{Y}}\mathbf{X-}\left[ \mathbf{X},\mathbf{Y}\right] .$
\end{theorem}

For a fixed $\mathbf{N,}$ both linear connections $\nabla $ and $\mathbf{D}$
are completely defined by the same metric $\mathbf{g}$.\footnote{%
In a series of our works, we wrote $\widehat{\mathbf{D}}$ for the canonical
d--connection and used "hats" for the coefficients and values computed for
this linear connection. In this article we shall use only the symbol $%
\mathbf{D;}$ in next sections, $\widehat{\mathbf{D}}$ will be used for a
conformal transforation of $\mathbf{D.}$} There is a substantial difference
between the canonical d--connection $\mathbf{D}$ and the connections used in
Riemann--Cartan geometry (see, for instance, \cite{penrr} and detailed
discussions with respect Einstein and Finsler geometries and metric--affine
generalizations in \cite{vdefq,vrflg,vsgg}). The torsion $\mathcal{T}$ is
completely defined by data $(\mathbf{g,N}),$ i.e. by the metric structure if
the value $\mathbf{N}$ is prescribed, but in the Einstein--Cartan gravity we
need additional algebraic equations for torsion.

\begin{corollary}
\label{corol1} There is a canonical distortion relation
\begin{equation}
\mathbf{D}=\nabla +\mathbf{Q},  \label{distrel}
\end{equation}%
where both linear connections $\ \mathbf{D}$ and $\nabla $ and the
distortion tensor $\mathbf{Q}$ are completely defined by the metric tensor $%
\mathbf{g}$ \ for a prescribed N--connection splitting.
\end{corollary}

Using the torsion tensor $\mathcal{R}$ (\ref{curv}) of $\mathbf{D},$ we can
introduce in standard form the Ricci d--tensor $\ \mathbf{Ric}$, which is
nonsymmetric because of nontrivial nonholonomic/ torsion structure, the
curvature scalar $\ \ _{s}R$ and the Einstein d--tensor $\mathbf{E}.$ The
(canonical nonholonomic) Einstein equations for $\ \mathbf{D}$ are written
geometrically
\begin{equation}
\mathbf{E}=\mathbf{Ric}-\frac{1}{2}\mathbf{g}\ _{s}R=\ \Upsilon ,
\label{deinsteq}
\end{equation}%
where the source $\Upsilon $ can be computed as in general relativity (GR)
with spacetime metric but $\nabla \rightarrow \mathbf{D}.$ A N--adapted
variational calculus with possible matter fields (fluids, bosons, fermions
etc, all with respect to nonholonomic frames) can be formulated but this
will result in nonsymmetric d--tensors $\Upsilon $, which is not surprising
because of nonholonmic character of such constructions with induced torsion.
On (pseudo) Riemannian manifolds with singnature $(-+++)$ and a prescribed
nonintegrable decomposition $2+2$, we can perform local constructions with $%
\ \Upsilon \rightarrow \{T_{\alpha \beta }\},$ where $T_{\alpha \beta }$ is
the energy--momentum tensor in GR. This allows us to provide a physical
interpretation to interactions constants. We need additional assumptions on
new interaction constants if the equations are considered on a tangent
bundle with total dimension 8; from a formal geometric point of view, there
are certain canonical lifts of geometric objects on Lorentz manifolds to
their tangent bundles.

The canonical distortion tensor $\mathbf{Q}[\mathbf{g,N}]$ from (\ref%
{distrel}) is an algebraic combination of the nonholonomically induced
torsion $\mathcal{T}$ $[\mathbf{g,N}]$ all completely defined by values $%
\mathbf{g}$ and $\mathbf{N.}$ It is possible to choose certain integrable
configurations $^{int}\mathbf{N}$ (this is equivalent to transforms of
geometric constructions with respect to certain classes of locally
integrable frames of reference) when
\begin{equation}
\Omega =0,\ \mathcal{T}=0,\ \mathbf{Q}=0  \label{lccond}
\end{equation}

\begin{theorem}
\label{thneq}For integrable N--connection structures when $\mathbf{D}_{\mid
\mathcal{T}=0}=\nabla ,$ the canonical nonholonomic Einstein equations (\ref%
{deinsteq}) for (pseudo) Riemannian metrics of dimension $2+2$ are
equivalent to the Einstein equations in GR.
\end{theorem}

\begin{proof}
It is a straightforward consequence of Corollary \ref{corol1} and (\ref%
{lccond}) and above presented considerations on sources.

$\square $ (end proof).
\end{proof}

\subsubsection{Formulas in N--adapted frames and coordinates}

We shall denote the local coordinates on a nonholonomic manifold $\mathbf{V}$
in the form $u=(x,y),$ or $u^{\alpha }=(x^{i},y^{a}),$ where the h--indices
run values $i,j,...=1,2,...n$ (for nonholonomic deformations in GR, $%
i,j,...=1,2$ or, on tangent to Lorentz bundles, $i,j,...=1,2,3,4$) and the
v--indices take values $a,b,c,...=n+1,n+2,n+m$ (for nonholonomic
deformations in GR, $a,b,...=3,4$ and, on tangent to Lorentz bundles, or $%
a,b,...=5,6,7,8$). For bundle spaces, $y^{a}$ are typical fiber coordinates
and $x^{i}$ are coordinates on base manifolds. We can introduce on $\mathbf{V%
}$ certain local coordinate bases $\partial _{\underline{\alpha }}=\partial
/\partial u^{\underline{\alpha }}=(\partial _{i}=\partial /\partial
x^{i},\partial _{a}=\partial /\partial y^{a})$ and their duals $du^{%
\underline{\beta }}=(dx^{j},dy^{b})$ [we shall emphasize some indices if it
is necessary that they are coordinate ones but omit "underlining" when that
will not result in ambiguities].

Transforms to arbitrary local frames, $e_{\alpha },$ and (co) frames, $%
e^{\beta },$ are given by nondegenerate "vierbein" matrices, $e_{\ \alpha }^{%
\underline{\alpha }}(u),$ and their duals, $e_{\ \underline{\beta }}^{\beta
}(u),$ respectively, $e_{\alpha }=e_{\ \alpha }^{\underline{\alpha }%
}\partial _{\underline{\alpha }}$ and $e^{\beta }=e_{\ \underline{\beta }%
}^{\beta }du^{\underline{\beta }}.$ Such transforms do not preserve a
N--connection splitting and mix $h$-$v$--indices. "Not--underlined" indices $%
\alpha ,\beta ,...;i,j,...;a,b,...$ will be considered, in general, as
abstract labels \cite{penrr}. The indices may be considered as coordinate
ones for decompositions with respect to coordinate bases (in our works, we
do not consider "boldface" indices but only "boldface" symbols for
spaces/geometric objects enabled with/adapted to a N--connection structure).

Locally, a N--connection $\mathbf{N}$ (\ref{whitney}) is defined by its
coefficients $N_{i}^{a}(u),$%
\begin{equation*}
\mathbf{N}=N_{i}^{a}(u)dx^{i}\otimes \partial /\partial y^{a}.
\end{equation*}

\begin{proposition}
A N--connection structure states a N--linear system of reference,%
\begin{equation}
\mathbf{e}_{\alpha }=\left( \mathbf{e}_{i}=\frac{\partial }{\partial x^{i}}%
-N_{i}^{b}\frac{\partial }{\partial y^{b}},e_{a}=\frac{\partial }{\partial
y^{a}}\right),  \label{dder}
\end{equation}%
and its dual%
\begin{equation}
\mathbf{e}^{\beta }=\left( e^{j}=dx^{j},\mathbf{e}%
^{b}=dy^{b}+N_{i}^{b}dx^{i}\right) .  \label{ddif}
\end{equation}
\end{proposition}

\begin{proof}
This follows from the possibility to construct N--adapted bases of type
\begin{equation}
\mathbf{e}_{\alpha }=\mathbf{e}_{\alpha }^{\ \underline{\alpha }}\partial _{%
\underline{\alpha }}\mbox{ and }\mathbf{e}_{\ }^{\beta }=\mathbf{e}_{\
\underline{\beta }}^{\beta }du^{\underline{\beta }},  \label{nadtetr}
\end{equation}%
where

\begin{equation}
\mathbf{e}_{\alpha }^{\ \underline{\alpha }}(u)=\left[
\begin{array}{cc}
e_{i}^{\ \underline{i}}(u) & N_{i}^{b}(u)e_{b}^{\ \underline{a}}(u) \\
0 & e_{a}^{\ \underline{a}}(u)%
\end{array}%
\right] ,~\mathbf{e}_{\ \underline{\beta }}^{\beta }(u)=\left[
\begin{array}{cc}
e_{\ \underline{i}}^{i\ }(u) & -N_{k}^{b}(u)e_{\ \underline{i}}^{k\ }(u) \\
0 & e_{\ \underline{a}}^{a\ }(u)%
\end{array}%
\right] .  \label{vbt}
\end{equation}%
$\square $
\end{proof}

One of the arguments to say that manifolds/bundles enabled with
N--connection structure are nonholonomic is that the frames (\ref{ddif})
satisfy the nonholonomy relations
\begin{equation}
\lbrack \mathbf{e}_{\alpha },\mathbf{e}_{\beta }]=\mathbf{e}_{\alpha }%
\mathbf{e}_{\beta }-\mathbf{e}_{\beta }\mathbf{e}_{\alpha }=W_{\alpha \beta
}^{\gamma }\mathbf{e}_{\gamma },  \label{anhrel}
\end{equation}%
where the (antisymmetric) nontrivial anholonomy coefficients are computed $%
W_{ia}^{b}=\partial _{a}N_{i}^{b}$ and $W_{ji}^{a}=\Omega _{ij}^{a}.$

\begin{proposition}
Any metric structure $\mathbf{g}$ on $\mathbf{V}$ can be written in
N--adapted form as a distinguished metric (d--metric)
\begin{equation}
\mathbf{g}=~^{h}g+~^{v}h=\ g_{ij}(u)\ e^{i}\otimes e^{j}+\ h_{ab}(u)\
\mathbf{e}^{a}\otimes \mathbf{e}^{b}.  \label{m1}
\end{equation}
\end{proposition}

\begin{proof}
Via frame/coordinate transforms, $g_{\alpha \beta }=e_{\ \alpha }^{\alpha
^{\prime }}e_{\ \beta }^{\beta ^{\prime }}g_{\alpha ^{\prime }\beta ^{\prime
}},$ any metric
\end{proof}

\begin{equation}
\mathbf{\ g}=\underline{g}_{\alpha \beta }\left( u\right) du^{\alpha
}\otimes du^{\beta }  \label{metr}
\end{equation}%
can written in the form
\begin{equation}
\underline{g}_{\alpha \beta }=\left[
\begin{array}{cc}
g_{ij}+N_{i}^{a}N_{j}^{b}h_{ab} & N_{j}^{e}h_{ae} \\
N_{i}^{e}h_{be} & h_{ab}%
\end{array}%
\right] .  \label{ansatz}
\end{equation}%
Introducing formulas (\ref{ddif}) and (\ref{vbt}) into (\ref{m1}) we obtain
the coordinate form (\ref{metr}) \ and (\ref{ansatz}). Inverse transforms
are similar.

$\square $

Using the last two propositions, we can compute the N--adapted coefficients $%
\ \mathbf{\Gamma }_{\ \alpha \beta }^{\gamma }=\left(
L_{jk}^{i},L_{bk}^{a},C_{jc}^{i},C_{bc}^{a}\right) ,$ with respect to frames
(\ref{dder}) and (\ref{ddif}), of the canonical d--connection $\mathbf{D},$
\begin{eqnarray}
L_{jk}^{i} &=&\frac{1}{2}g^{ir}\left( \mathbf{e}_{k}g_{jr}+\mathbf{e}%
_{j}g_{kr}-\mathbf{e}_{r}g_{jk}\right) ,  \label{candcon} \\
L_{bk}^{a} &=&e_{b}(N_{k}^{a})+\frac{1}{2}h^{ac}\left( e_{k}h_{bc}-h_{dc}\
e_{b}N_{k}^{d}-h_{db}\ e_{c}N_{k}^{d}\right) ,  \notag \\
C_{jc}^{i} &=&\frac{1}{2}g^{ik}e_{c}g_{jk},C_{bc}^{a}=\frac{1}{2}%
h^{ad}\left( e_{c}h_{bd}+e_{c}h_{cd}-e_{d}h_{bc}\right) .  \notag
\end{eqnarray}%
The N--adapted coefficients of d--torsion \newline
$\mathcal{T}$ $=\{\mathbf{T}_{\ \alpha \beta }^{\gamma }=(T_{\ jk}^{i},T_{\
ja}^{i},T_{\ ji}^{a},T_{\ bi}^{a},T_{\ bc}^{a})\}$\ (\ref{dtors1}) of $%
\mathbf{D}$ are computed
\begin{eqnarray}
T_{\ jk}^{i} &=&L_{\ jk}^{i}-L_{\ kj}^{i}=0,\ \ T_{\ ja}^{i}=-\ T_{\
aj}^{i}=C_{\ ja}^{i},\ T_{\ ji}^{a}=\Omega _{\ ji}^{a},\   \notag \\
\ T_{\ bi}^{a} &=&\partial _{b}N_{i}^{a}-L_{\ bi}^{a},T_{\ bc}^{a}=C_{\
bc}^{a}-C_{\ cb}^{a}=0.  \label{dtors}
\end{eqnarray}

We provide also the N--adapted coefficients of d--curvature \newline
$\mathcal{R}=\{\mathbf{R}_{\ \alpha \beta \nu }^{\tau }=(R_{\ hjk}^{i},R_{\
bjk}^{a},R_{\ jka}^{i},R_{\ bka}^{c},R_{\ jbc}^{i},R_{\ bcd}^{a})\}$ (\ref%
{curv}) of $\mathbf{D},$ {\small
\begin{eqnarray}
R_{\ hjk}^{i} &=&\mathbf{e}_{k}\ L_{\ hj}^{i}-\mathbf{e}_{j}\ L_{\
hk}^{i}+L_{\ hj}^{m}L_{\ mk}^{i}-L_{\ hk}^{m}\ L_{\ mj}^{i}-C_{\
ha}^{i}\Omega _{\ kj}^{a},  \notag \\
R_{\ bjk}^{a} &=&\mathbf{e}_{k}\ L_{\ bj}^{a}-\mathbf{e}_{j}\ L_{\
bk}^{a}+L_{\ bj}^{c}\ L_{\ ck}^{a}-L_{\ bk}^{c}L_{\ cj}^{a}-C_{\
bc}^{a}\Omega _{\ kj}^{c},  \notag \\
R_{\ jka}^{i} &=&e_{a}\ L_{\ jk}^{i}-\ D_{k}C_{\ ja}^{i}+C_{\ jb}^{i}\ T_{\
ka}^{b},  \label{dcurv} \\
\ R_{\ bka}^{c} &=&e_{a}L_{\ bk}^{c}-D_{k}\ C_{\ ba}^{c}+C_{\ bd}^{c}T_{\
ka}^{c},  \notag \\
R_{\ jbc}^{i} &=&e_{c}\ C_{\ jb}^{i}-e_{b}\ C_{\ jc}^{i}+C_{\ jb}^{h}C_{\
hc}^{i}-C_{\ jc}^{h}C_{\ hb}^{i},  \notag \\
R_{\ bcd}^{a} &=&e_{d}C_{\ bc}^{a}-e_{c}C_{\ bd}^{a}+C_{\ bc}^{e}C_{\
ed}^{a}-C_{\ bd}^{e}\ C_{\ ec}^{a}.  \notag
\end{eqnarray}%
} Contracting indices, we can compute the h- v--components $\mathbf{R}%
_{\alpha \beta }\doteqdot \mathbf{R}_{\ \alpha \beta \tau }^{\tau }$ of the
Ricci tensor $\mathbf{Ric}$,
\begin{equation}
R_{ij}\doteqdot R_{\ ijk}^{k},\ R_{ia}\doteqdot -R_{\
ika}^{k},R_{ai}\doteqdot R_{\ aib}^{b},\ R_{ab}\doteqdot R_{\ abc}^{c}.
\label{dricci}
\end{equation}%
The scalar curvature is
\begin{equation}
\ _{s}R\doteqdot \mathbf{g}^{\alpha \beta }\mathbf{R}_{\alpha \beta
}=g^{ij}\ R_{ij}+h^{ab}R_{ab}.  \label{sdccurv}
\end{equation}

In component form, the analog of Theorem \ref{thneq} is

\begin{theorem}
The Einstein equations in GR are equivalent to
\begin{eqnarray}
\mathbf{R}_{\ \beta \delta }-\frac{1}{2}\mathbf{g}_{\beta \delta }\ \ _{s}R
&=&\mathbf{\Upsilon }_{\beta \delta },  \label{cdeinst} \\
L_{aj}^{c}=e_{a}(N_{j}^{c}),\ C_{jb}^{i}=0,\ \Omega _{\ ji}^{a} &=&0,
\label{lcconstr}
\end{eqnarray}%
written for the canonical d--connection coefficients (\ref{candcon}) if \ $%
\mathbf{\Upsilon }_{\beta \delta }\rightarrow T_{\beta \delta }$
(energy--momentum tensor for matter) for $\mathbf{D}\rightarrow \nabla .$
\end{theorem}

\begin{proof}
It follows from above component formulas introduced in (\ref{deinsteq}) and (%
\ref{lccond}). The constraints (\ref{lcconstr}) are equivalent to (\ref%
{lccond}), i.e. to the condition of zero torsion (\ref{dtors}) and zero
distortion d--tensors $\mathbf{Q}=0,$ which results in $\ \mathbf{D=}\nabla
, $ see (\ref{distrel}).

$\square $
\end{proof}

The main reason to work with equations of type (\ref{deinsteq}) and (\ref%
{cdeinst}) is that such equations for $\mathbf{D}$ (we say "in nonholonomic
variables") decouple with respect to N--adapted frames (for spaces with
splitting of dimension $2,$ or $3,+2+2+2+...$) for generic off--diagonal
ansatz for metric $\mathbf{g}$ and certain parameterizations of $\mathbf{N}$
depending on all coordinates. This allows us to integrate such nonlinear PDE
in very general forms. We construct integral varieties determined by
corresponding classes of generating and integration functions and
integration constant which may be defined from certain boundary/Cauchy
conditions and additional physical arguments. Imposing additional
Levi--Civita (LC) conditions (\ref{lccond}), which constrain
nonholonomically the integral varieties of solutions of $\ \mathbf{E}%
=\Upsilon $, we can "extract" solutions in GR. We note that we can not
decouple and integrate in such off--diagonal forms the Einstein equations if
we work from the very beginning and only with $\nabla .$ The main "trick" is
that we "relax" the constraints of zero torsion in the standard Einstein
equations by considering an "auxiliary" connection $\mathbf{D}$ (in next
section, we shall see that this is a Finsler type d--connection); such
constructions are provided in Refs. \cite{veyms,vsgg,vrflg}.

\subsection{Metric compatible Finsler--Cartan geometries}

We outline some results from the Finsler geometry on tangent bundles \cite%
{bejancu1,cartf,rund,matsumoto,bcs} and show how the constructions can be
re--defined for nonholonomic (pseudo) Riemannian manifolds \cite%
{vrflg,vsgg,vcrit}.

\subsubsection{The Finsler fundamental/generating function}

Let us consider a tangent bundle $TM=\bigcup\nolimits_{x\in M}T_{x}M,$ where
$T_{x}M$ are the tangent spaces at points $x\in M,$ for the base space $M$
being a real $\mathit{C}^{\infty }$ manifold of dimension $\dim M=n.$
Roughly, the term Finsler \textquotedblright metric\textquotedblright\ $F$
is used for a (Finsler) geometry determined on $TM$ by a nonlinear quadratic
element%
\begin{equation}
ds^{2}=F^{2}(x,dx),  \label{nqe}
\end{equation}%
when $dx^{i}\sim y^{i}.$ This generalizes the well--known and very important
example of (pseudo) Riemannian geometry, determined by a metric tensor $%
g_{ij}(x^{k}).$ Taking  a particular case with quadratic form $F=\sqrt{%
|g_{ij}(x)y^{i}y^{j}|}$ we obtain
\begin{equation}
ds^{2}=g_{ij}(x)dx^{i}dx^{j}.  \label{lqe}
\end{equation}%
Such an element states a geometry on $M$ with geometric objects depending
only on $x$--variables even for definitions of tensors, linear connections,
spinors etc objects the tangent bundle $TM$ is also involved in order to
define such objects by analogy to flat spaces.

\begin{definition}
\label{fgf}A Finsler fundamental/generating function (metric) is a function $%
F:\ TM\rightarrow \lbrack 0,\infty )$ subjected to the conditions:

\begin{enumerate}
\item $F(x,y)$ is $\mathit{C}^{\infty }$ on $\widetilde{TM}:=TM\backslash
\{0\},$ for $\{0\}$ denoting the set of zero sections of $TM$ on $M;$

\item $F(x,\beta y)=\beta F(x,y),$ for any $\beta >0,$ i.e. it is a positive
1--homogeneous function on the fibers of $TM;$

\item for any $y\in \widetilde{T_{x}M},$ the Hessian
\begin{equation}
\ ^{v}{\tilde{g}}_{ij}(x,y)=\frac{1}{2}\frac{\partial ^{2}F^{2}}{\partial
y^{i}\partial y^{j}}  \label{hessian}
\end{equation}%
is considered as s a \textquotedblright vertical\textquotedblright\ (v)
metric on typical fiber, i.e. it is nondegenerate and positive definite, $%
\det |\ ^{v}{\tilde{g}}_{ij}|\neq 0.$
\end{enumerate}
\end{definition}

If the base $M$ is taken to be a Lorentz manifold in GR, we can construct
generalizations on $TM$ with a good physical axiomatic system which is very
similar to that of Einstein gravity when the Levi--Civita connection $\nabla
$ is substituted by a metric compatible Finsler variant of the canonical
d--connection ${\mathbf{D}},$ see discussions in \cite{vcosm} and
next subsections.

\begin{remark}
\label{rpfl}The condition 3 above should be relaxed to "not positive
definite" for models of Finsler gravity with finite, in general, locally
anisotropic speed of light.
\end{remark}

Considering a background (pseudo) Riemannian metric $g_{ij}(x)$ with
signature $(+,+,+,-)$ on $M,$ we can elaborate various geometric and
physical models on $TM$ with locally anisotropic metrics $\ g_{ij}(x,y)$
depending on "velocity" type coordinates $y^{a}.$ The main difference
between (pseudo) Riemannian and Finsler geometries is that the first type
ones are completely defined by a metric structure $g_{ij}(x)$ (from which a
unique Levi--Civita connection $\nabla $ can be constructed) but the second
type ones can not be completely derived from a Finsler metric $F(x,y)$
and/or its Hessian $\ ^{v}{\tilde{g}}_{ij}(x,y).$

\begin{remark}
\label{remarkfg} A complete Finsler geometry model $(F:\ _{F}\mathbf{N},\ \
_{F}\mathbf{g},\ \ _{F}\mathbf{D})$ can be defined by additional assumptions
on how three fundamental geometric objects (the N--connection $\ \ _{F}%
\mathbf{N}$, the total metric $\ \ _{F}\mathbf{g,}$ the d--connection $\ \
_{F}\mathbf{D}$) \ can be determined uniquely by a fundamental Finsler
function $F.$
\end{remark}

Finsler like geometries can be elaborated on a generic nonholonomic
bundle/manifold $\mathbf{V}$ following self--consistent geometric and
physically important principles (for instance, $\mathbf{V}=TM,$ $\mathbf{V}$
is a (pseudo) Riemannian manifold with nonholonomic $2+2$ splitting \cite%
{veyms}; there were performed similar generalizations for
supermanifolds/superbundles and/or noncommutative generalizations,
affine--Finsler spaces etc, see \cite{vsuperstr,vfncr,vsgg}).

\subsubsection{The canonical Finsler connections and lifts of
metrics}

Let us consider $L=F^{2}$ is considered as an effective regular Lagrangian
on $TM$ and action integral $S(\tau )=\int\limits_{0}^{1}L(x(\tau ),y(\tau
))d\tau$, for $y^{k}(\tau )=dx^{k}(\tau )/d\tau ,$ where $x(\tau )$
parameterizes smooth curves on a manifold $M$ with $\tau \in \lbrack 0,1].$

\begin{lemma}
The Euler--Lagrange equations $\frac{d}{d\tau }\frac{\partial L}{\partial
y^{i}}-\frac{\partial L}{\partial x^{i}}=0$ are equivalent to the
\textquotedblright nonlinear geodesic\textquotedblright\ (equivalently,
semi--spray) equations $\frac{d^{2}x^{k}}{d\tau ^{2}}+2\tilde{G}^{k}(x,y)=0$%
, where
\begin{equation}
\tilde{G}^{k}=\frac{1}{4}\tilde{g}^{kj}\left( y^{i}\frac{\partial ^{2}L}{%
\partial y^{j}\partial x^{i}}-\frac{\partial L}{\partial x^{j}}\right) ,
\label{smspr}
\end{equation}%
for $\tilde{g}^{kj}$ being inverse to $\ ^{v}{\tilde{g}}_{ij}\equiv {\tilde{g%
}}_{ij}$ (\ref{hessian}).
\end{lemma}

Certain geometric properties of fundamental Finsler functions can be studied
via semi--spray configurations not concerning the problem of definition of
connections and metrics for such spaces. For instance, J. Kern \cite{kern}
suggested to consider nonhomomgeneous regular Lagrangians instead of those
considered in Finsler geometry. That resulted in so--called
Lagrange--Finsler geometry, on applications
in modern physics see \cite{vrflg,vcrit}.

\begin{definition}
\textbf{-Corollary:}\ There is a canonical N--connection $\mathbf{\tilde{N}}%
=\{\tilde{N}_{j}^{a}\},$ $\ $%
\begin{equation}
\tilde{N}_{j}^{a}:=\frac{\partial \tilde{G}^{a}(x,y)}{\partial y^{j}},
\label{cncl}
\end{equation}%
completely defined by the fundamental Finsler function $F.$
\end{definition}

\begin{proof}
Using the above Lemma and local computations we can verify that the
conditions Definition--Theorem (\ref{deftnc}) for N--connections are
satisfied. See also details of such a proof in \cite{vsgg}.

$\square $
\end{proof}

We note that via $\mathbf{\tilde{N}}$ a Finsler metric $F$ defines naturally
certain N--adapted frame structures $\mathbf{\tilde{e}}_{\nu }=(\mathbf{%
\tilde{e}}_{i},e_{a})$ and $\mathbf{\tilde{e}}^{\mu }=(e^{i},\mathbf{\tilde{e%
}}^{a}):$ we have to substitute $N_{j}^{a}\rightarrow \tilde{N}_{j}^{a}$
into, respectively, (\ref{dder}) and (\ref{ddif}).

\begin{definition}
\textbf{-Corollary:}\ A total metric structure on $TM$ can be defined by a
Sasaky type lift of ${\tilde{g}}_{ij},$%
\begin{equation}
\mathbf{\tilde{g}}=\tilde{g}_{ij}(x,y)\ e^{i}\otimes e^{j}+\tilde{g}%
_{ij}(x,y)\ \mathbf{\tilde{e}}^{i}\otimes \ \mathbf{\tilde{e}}^{j}.
\label{slm}
\end{equation}
\end{definition}

It is possible to use other geometric principles for \textquotedblright
lifts and projections\textquotedblright\ when, for instance, from a given $F$
it is constructed a complete homogeneous metric on total/horizontal spaces
of $TM$. For models of locally anisotropic/Finsler gravity on $TM,$ or on $%
\mathbf{V,}$ a generalized covariance principle has to be considered
following geometric and physical considerations \cite{vcosm}. Such
constructions are performed up to certain frame/coordinate transforms $%
\mathbf{\tilde{e}}_{\gamma }\rightarrow \mathbf{e}_{\gamma ^{\prime }}=e_{\
\gamma ^{\prime }}^{\gamma }\mathbf{\tilde{e}}_{\gamma }.$ From a formal
point of view, we can omit \textquotedblright tilde\textquotedblright\ on
symbols and write, in general, $\mathbf{g=\{g}_{\alpha \beta }$ and $\mathbf{%
N=\{}N_{i}^{a}=e_{\ a^{\prime }}^{a}e_{i}^{\ i^{\prime }}N_{i^{\prime
}}^{a^{\prime }}\}.$ We can define a subclass of frame/coordinate transforms
preserving a prescribed splitting (\ref{whitney}).

\subsubsection{Models of Finsler--Cartan spaces}

\label{ssmfcs}Using last two Definition--Corollaries, we prove

\begin{theorem}
A fundamental Finsler function $F(x,y)$ defines naturally a nonholonomic
Riemann--Cartan model on $\widetilde{TM}$ determined by geometric data $(F:\
_{F}\mathbf{N=\tilde{N}},\ \ _{F}\mathbf{g=\tilde{g}},\ \ _{F}\mathbf{D=}\ \
\mathbf{D})$, where $\ \mathbf{D}$ is determined by N--adapted coefficients $%
\ \mathbf{\Gamma }_{\ \alpha \beta }^{\gamma }=(\
L_{jk}^{i},L_{bk}^{a},C_{jc}^{i},C_{bc}^{a})$ computed using formulas (\ref%
{candcon}) for $\mathbf{g\rightarrow \tilde{g}}$ (\ref{slm}) and $\mathbf{%
N\rightarrow \tilde{N}}$ (\ref{cncl}).
\end{theorem}

Introducing coefficients $\ \ \mathbf{\Gamma }_{\ \alpha \beta }^{\gamma },$
respectively, into formulas (\ref{dtors}) and (\ref{dcurv}), we compute the
torsion $\ \ \mathcal{T}$ $\ $and curvature $\ \ \mathcal{R}$ of $\ \
\mathbf{D}.$

In Finsler geometry it is largely used the Cartan d--connection $\mathbf{%
\tilde{D}}$ \cite{cartf}, see details in \cite{rund}, which is also
metric compatible and can be related to $\ \mathbf{D}$ (\ref{candcon}) via
frame transforms and deformations. If we consider that $\
L_{bk}^{a}\rightarrow \ L_{jk}^{i}$ and $\ C_{jc}^{i}\rightarrow $ $\ \
C_{bc}^{a}$ for arbitrary $\mathbf{g}$ and $\mathbf{N}$ on $TM$ (i.e. we
identify respectively $a=n+i$ with $i$ and $b=n+j),$ we obtain the
so--called normal d--connection \ $\ ^{n}\mathbf{D}=(\ \ \ ^{n}L_{\
jk}^{i},\ \ \ ^{n}C_{jc}^{i})$ where
\begin{equation}
\ \ \ ^{n}L_{\ jk}^{i}=\frac{1}{2}g^{ih}(\mathbf{e}_{k}g_{jh}+\mathbf{e}%
_{j}g_{kh}-\mathbf{e}_{h}g_{jk}),\ \ \ ^{n}C_{\ bc}^{a}=\frac{1}{2}%
g^{ae}(e_{b}h_{ec}+e_{c}h_{eb}-e_{e}h_{bc}).  \label{cdc}
\end{equation}

\begin{definition}
The Cartan d--connection \ $\mathbf{\tilde{D}}=(\tilde{L}_{\ jk}^{i},\tilde{C%
}_{jc}^{i})$ is defined by introducing $\mathbf{g=\tilde{g}}$ with $\tilde{h}%
_{ij}=\tilde{g}_{ij}$\ and $\mathbf{N=\tilde{N}}$ in (\ref{cdc}).
\end{definition}

Using formulas (\ref{dtors}) and (\ref{dcurv}) for N--adapted coefficients
of $\mathbf{\tilde{D},}$ we prove

\begin{theorem}
The nontrivial components of torsion $\mathbf{\tilde{T}}_{\beta \gamma
}^{\alpha }=\{\tilde{T}_{jc}^{i},\tilde{T}_{ij}^{a},\tilde{T}_{ib}^{a}\}$
and curvature\newline
 $\mathbf{\tilde{R}}_{\ \beta \gamma \tau }^{\alpha }=\{\tilde{R%
}_{\ hjk}^{i},\tilde{P}_{\ jka}^{i},\tilde{S}_{\ bcd}^{a}\}$ of $\mathbf{%
\tilde{D}}$ are respectively {\small
\begin{equation}
\tilde{T}_{jk}^{i}=0,\tilde{T}_{jc}^{i}=\tilde{C}_{\ jc}^{i},\tilde{T}%
_{ij}^{a}=\tilde{\Omega}_{ij}^{a},\tilde{T}_{ib}^{a}=e_{b}\left( \tilde{N}%
_{i}^{a}\right) -\tilde{L}_{\ bi}^{a},\tilde{T}_{bc}^{a}=0,  \label{torscdc}
\end{equation}%
and
\begin{eqnarray}
\tilde{R}_{\ hjk}^{i} &=&\mathbf{\tilde{e}}_{k}\tilde{L}_{\ hj}^{i}-\mathbf{%
\tilde{e}}_{j}\tilde{L}_{\ hk}^{i}+\tilde{L}_{\ hj}^{m}\tilde{L}_{\ mk}^{i}-%
\tilde{L}_{\ hk}^{m}\tilde{L}_{\ mj}^{i}-\tilde{C}_{\ ha}^{i}\tilde{\Omega}%
_{\ kj}^{a},  \label{curvcart} \\
\tilde{P}_{\ jka}^{i} &=&e_{a}\tilde{L}_{\ jk}^{i}-\mathbf{\tilde{D}}_{k}%
\tilde{C}_{\ ja}^{i},\ \tilde{S}_{\ bcd}^{a}=e_{d}\tilde{C}_{\ bc}^{a}-e_{c}%
\tilde{C}_{\ bd}^{a}+\tilde{C}_{\ bc}^{e}\tilde{C}_{\ ed}^{a}-\tilde{C}_{\
bd}^{e}\tilde{C}_{\ ec}^{a}.  \notag
\end{eqnarray}
}
\end{theorem}

A very important property of $\left( \mathbf{\tilde{g};}\ \tilde{h}_{ij}=%
\tilde{g}_{ij},\mathbf{\tilde{D}}\right) $ is that such geometric data can
be encoded equivalently into an almost K\"{a}hler structure \cite{matsumoto}.
  This allows us to perform
deformation quantization and or A--brane quantization of Finsler geometry
and generalizations, see \cite{vdefq,vbraneq}. Such constructions are
important for definition of almost K\"{a}hler spinors and Dirac operators in
Finsler geometry (we do not present details in this work but emphasize that
almost symplectic Finsler structures can be encoded into corresponding
spinor and twistor structures).

\subsubsection{On metric noncompatible Finsler geometries}

Mathematicians elaborated different models of Finsler geometry generated by
a fundamental Finsler function $F(x,y).$ Most known are constructions due to
L. Berwald \cite{berwald} and S. Chern \cite{chern} (see details in \cite%
{bcs}) and "nonstandard" definition for the Ricci curvature by H.
Akbar--Zadeh \cite{akbar}. For instance,

\begin{itemize}
\item the Berwald d--connection is $\ ^{B}\mathbf{D:}=(\ ^{B}L_{\
jk}^{i}=\partial \tilde{N}_{j}^{i}/\partial y^{k},\ ^{B}C_{jc}^{i}=0);$

\item the Chern d--connection is $\ ^{Ch}\mathbf{D:}=(\ ^{Ch}L_{\ jk}^{i}=%
\tilde{L}_{\ jk}^{i},\ ^{Ch}C_{jc}^{i}=0).$
\end{itemize}

The Chern's d--connection is very similar to the Levi--Civita connection,
for geometric constructions on the $h$--subspace. The Finsler geometries
determined by such d--connections are not metric compatible on total space
of $TM$ and characterized by nontrivial nonmetricity fields, $\mathcal{Q}:=%
\mathbf{Dg,}$$\ ^{B}\mathcal{Q}\neq 0$ and $\ ^{Ch}\mathcal{Q}\neq 0$. We
studied various generalizations affine--Finsler and affine--Lagrange spaces
in Part I of \cite{vsgg}. Nontrivial nonmetricity fields (and "nonstandard"
definitions of scalar and Ricci curvatures of Finsler spaces) present, in
general, difficulties for definition of spinors and Dirac type operators,
formulating conservation laws etc, see critical remarks in \cite%
{vcrit,vrflg}. So, there are substantial geometric and physical reasons
to work with Finsler--Cartan type spaces and similar metric compatible
configurations for applications in modern gravity and cosmology.

\subsubsection{Finsler variables in general relativity}

In this section, we show how the Einstein gravity can re--written
equivalently in Finsler like variables.

Let us consider a (pseudo) Riemannian space $\mathbf{V}$ with nonholonomic
2+2 splitting $\mathbf{N}=\{N_{i}^{a}\}$ and d--metric $\mathbf{g}=\{\mathbf{%
g}_{\alpha \beta }\}=\{\underline{g}_{\alpha ^{\prime }\beta ^{\prime }}\},$
which can be written in the form (\ref{m1}) and/or (\ref{metr}) \ and (\ref%
{ansatz}). We can always introduce on a well--defined cart for an atlas
covering $\mathbf{V}$ a homogeneous function $\mathcal{F}(x,y)$ satisfying
the conditions of Definition \ref{fgf} and Remark \ref{rpfl}. Using such a
formal (pseudo) Finsler generating function, we can construct a Sasaki
d--metric of type (\ref{slm}), for $\widetilde{f}_{ij}:=\frac{1}{2}\frac{%
\partial ^{2}\mathcal{F}^{2}}{\partial y^{i}\partial y^{j}}\mathcal{\ }$and $%
\widetilde{\mathcal{N}}_{j}^{a}$ obtained for $F\rightarrow \mathcal{F}$
following formulas (\ref{cncl}) and (\ref{smspr}). With respect to dual
local basis $du^{\alpha }=(dx^{i},dy^{a}),$ such a total metric can be
written in the form
\begin{equation*}
\underline{f}_{\alpha \beta }=\left[
\begin{array}{cc}
\widetilde{f}_{ij}+\widetilde{\mathcal{N}}_{i}^{a}\widetilde{\mathcal{N}}%
_{j}^{b}\widetilde{f}_{ab} & \widetilde{\mathcal{N}}_{j}^{e}\widetilde{f}%
_{ae} \\
\widetilde{\mathcal{N}}_{i}^{e}\widetilde{f}_{be} & \widetilde{f}_{ab}%
\end{array}%
\right] .
\end{equation*}%
Solving a quadratic algebraic equation for $e_{\ \alpha ^{\prime }}^{\alpha
}(u),$ for given values $\underline{g}_{\alpha ^{\prime }\beta ^{\prime }}$
and $\underline{f}_{\alpha \beta }(u),$
\begin{equation}
\underline{g}_{\alpha ^{\prime }\beta ^{\prime }}(u)=e_{\ \alpha ^{\prime
}}^{\alpha }(u)e_{\ \beta ^{\prime }}^{\beta }(u)\underline{f}_{\alpha \beta
}(u),  \label{ftransf}
\end{equation}%
we can re--write connections and tensors on $\mathbf{V,}$ up to
frame/coordinate transforms, in terms of variables $\left( \mathcal{F}:%
\widetilde{\mathbf{f}}\right) $ or $\left( \mathbf{g,N}\right) .$ We may
change the carts and coordinates and $\mathcal{F}$ in order to get real
well--defined solutions for vierbeins $e_{\ \alpha ^{\prime }}^{\alpha }.$

The above constructions depend on arbitrary generating function $\mathcal{F}%
, $ which states a 2+2 splitting via formulas (\ref{cncl}) and (\ref{smspr})
and respective frames (\ref{dder}) and (\ref{ddif}), in their turn admitting
transforms to N--elongated values determined by $N_{i}^{a}$ and/or $%
\widetilde{\mathcal{N}}_{j}^{a}.$ This reflects the principle of general
covariance when some additional nonholonomic constraints are imposed on
frame structure. If a relation (\ref{ftransf}) is established on $\mathbf{V,}
$ we can compute the Levi--Civita connection $\nabla $ using the values $%
\underline{f}_{\alpha \beta }$ and/or, equivalently, $\mathbf{g}_{\alpha
\beta }.$ We can also compute the coefficients of $\mathbf{D}$ (\ref{candcon}%
) and $\widetilde{\mathbf{D}}$ (\ref{cdc}) \ with distortion relations of
type (\ref{distrel}). $\ $All such values are completely determined by $%
\underline{g}_{\alpha ^{\prime }\beta ^{\prime }}$ (equivalently by $%
\underline{f}_{\alpha \beta }$). Technically, it is difficult to solve in
general form the Einstein equations of $\nabla $ written in Finsler like
variable because they contain terms up to forth derivatives of $\mathcal{F}$
etc. Nevertheless, we can use some convenient data $\left( \mathbf{g,N}%
\right) $ in order to find a general solution $\mathbf{g}_{\alpha \beta }$
of the system (\ref{cdeinst}) and to find some variables $\underline{f}%
_{\alpha \beta }$ using (\ref{ftransf}). If the constraints (\ref{lcconstr})
are imposed additionally, we generate solutions in GR. For Finsler
generalizations, we do not have to consider such Levi--Civita conditions.
\begin{conclusion}
\begin{enumerate}
\item Any metric compatible Finsler--Cartan geometry can be modelled as a
nonholonomic Riemann--Cartan geometry with an effective d--torsion completely
determined by the metric and N--connection structures. We do not need
additional algebraic equations as in Einstein--Cartan gravity in order to
find the d--torsion coefficients.

\item Any (pseudo) Riemannian manifold can be equivalently described by
geometric data $(\mathbf{g},\nabla )$, and/or $\left( \mathbf{g},\mathbf{N},%
\mathbf{D}\right) $, and/or, in Finsler like variables, $(\mathcal{F}:%
\widetilde{\mathbf{f}}=\mathbf{g},\widetilde{\mathcal{N}}_{j}^{a},\widetilde{%
\mathbf{D}}).$
\end{enumerate}
\end{conclusion}

\subsection{Conformal transforms and N--connections}

With respect to arbitrary or coordinate frames, it is not a trivial task to
define conformal transforms because of generic anisotropy of spaces enabled
with N--connection structure (in particular, for Finsler--Cartan spaces) and
nonlinear dependence of metric and connections on $N_{i}^{a}$ and/or $%
\widetilde{\mathcal{N}}_{j}^{a}.$ Nevertheless, in N--adapted frames (\ref%
{dder}) and (\ref{ddif}), certain analogy to Riemann--Cartan spaces can be
found.

Let us denote by $\mathbf{D}$ any of metric compatible d--connections (\ref%
{candcon}) or $\mathbf{\tilde{D}}$ (\ref{cdc}). The torsion and curvature
tensors (see N--adapted coefficients (\ref{dtors}) and (\ref{dcurv}) and,
respectively, (\ref{torscdc}) and (\ref{curvcart})) are computed in abstract
index form via%
\begin{equation*}
\mathbf{\Delta }_{\alpha \beta }f=\mathbf{T}_{\ \alpha \beta }^{\gamma }%
\mathbf{D}_{\gamma }f\mbox{ and }\left( \mathbf{\Delta }_{\alpha \beta }-%
\mathbf{T}_{\ \alpha \beta }^{\gamma }\mathbf{D}_{\gamma }\right) \mathbf{V}%
^{\tau }=\mathbf{R}_{\ \alpha \beta \gamma }^{\tau }\mathbf{V}^{\gamma },
\end{equation*}%
for
\begin{equation}
\mathbf{\Delta }_{\alpha \beta }:=\mathbf{D}_{\alpha }\mathbf{D}_{\beta }-%
\mathbf{D}_{\beta }\mathbf{D}_{\alpha }=2\mathbf{D}_{[\alpha }\mathbf{D}%
_{\beta ]}  \label{dcommut}
\end{equation}%
and arbitrary scalar function $f(x,y)$ and d--vector $\mathbf{V}^{\gamma }$
(in this work, we follow a different rule/order of contracting indices than
that in \cite{penrr}).

We can consider a source d--tensor $\mathbf{\Upsilon }_{\alpha \beta
}=-\lambda \mathbf{g}_{\alpha \beta }+8\pi G\mathbf{T}_{\alpha \beta },$
where, for $2+2$ splitting, $\lambda $ and $G$ are respectively the
cosmological and Newton constants (such values can be defined via Sasaki
lifts, for $4+4$ models on tangent bundles). The Einstein equations for $%
\mathbf{D}_{\alpha }$ can be written similarly to (\ref{cdeinst}),%
\begin{equation}
\mathbf{R}_{\alpha \beta }-\frac{1}{2}\mathbf{g}_{\alpha \beta }\ \
_{s}R+\lambda \mathbf{g}_{\alpha \beta }=8\pi G\mathbf{T}_{\alpha \beta },
\label{einst1}
\end{equation}%
where $\mathbf{R}_{\alpha \beta }:=\mathbf{R}_{\ \alpha \beta \gamma
}^{\gamma }$ and$\ _{s}R:=\mathbf{g}^{\alpha \beta }\mathbf{R}_{\alpha \beta
}.$ In the spinor formulation of gravity, there are used
\begin{eqnarray}
&&\ _{s}R:=24\Lambda =4\lambda -8\pi G\mathbf{T}_{\ \tau }^{\tau },
\label{einst2} \\
&&\mathbf{\Phi }_{\alpha \beta }:=3\Lambda \mathbf{g}_{\alpha \beta }-\frac{1%
}{2}\mathbf{R}_{\alpha \beta }=8\pi G(\frac{1}{4}\mathbf{T}_{\ \tau }^{\tau }%
\mathbf{g}_{\alpha \beta }-\mathbf{T}_{\alpha \beta })  \notag
\end{eqnarray}%
and the conformal d--tensor%
\begin{equation}
\mathbf{C}_{\ \alpha \beta }^{\tau \quad \gamma }:=\mathbf{R}_{\ \alpha
\beta }^{\tau \quad \gamma }+2\mathbf{R}_{[\alpha }^{\quad \lbrack \tau
}\delta _{\beta ]}^{\gamma ]}+\frac{1}{3}\ _{s}R\delta _{\lbrack \alpha
}^{\gamma }\delta _{\beta ]}^{\tau }=\mathbf{R}_{\ \alpha \beta }^{\tau
\quad \gamma }+4\mathbf{P}_{[\alpha }^{\quad \lbrack \tau }\delta _{\beta
]}^{\gamma ]}  \label{wtcd}
\end{equation}%
where $\delta _{\beta }^{\gamma }$ is the Kronecker symbol and
\begin{equation}
2\mathbf{P}_{\alpha \beta }=\frac{1}{6}\ _{s}R\mathbf{g}_{\alpha \beta }-%
\mathbf{R}_{\alpha \beta }.  \label{pdt}
\end{equation}
Such d--tensor formulas are related to similar ones for the Levi--Civita
connection $\nabla $ via distortions $\mathbf{D}=\nabla +\mathbf{Q}$ (\ref%
{distrel}),where all values are determined by a corresponding d--metric (\ref%
{m1}) or (\ref{slm}). This results in distortions of d--tensors,%
\begin{eqnarray}
\mathbf{R}_{\alpha \beta } &=&R_{\alpha \beta }+\mathbf{Q}_{\alpha \beta },\
_{s}R=R+\ _{s}Q,  \label{distweil} \\
\mathbf{R}_{\tau \alpha \beta \gamma } &=&R_{\tau \alpha \beta \gamma }+%
\mathbf{Q}_{\tau \alpha \beta \gamma },\ \mathbf{C}_{\tau \alpha \beta
\gamma }=C_{\tau \alpha \beta \gamma }+\ _{W}\mathbf{Q}_{\tau \alpha \beta
\gamma },  \notag
\end{eqnarray}%
were the left label $\ _{W}\mathbf{Q}_{\tau \alpha \beta \gamma }$ is from
the distortion of Weyl's type conformal d--tensor.

\begin{proposition}
\label{prop1}Under conformal transforms of coefficients d--metric (\ref{m1}%
),
\begin{equation}
\widehat{\mathbf{g}}_{\alpha \beta }:=\varpi ^{2}(u)\mathbf{g}_{\alpha \beta
},  \label{conftrdm}
\end{equation}%
preserving the N--connection structure $\mathbf{N}=\{N_{i}^{a}\},$ the
conformal d--tensor (\ref{wtcd}) satisfies the conditions%
\begin{equation*}
\widehat{\mathbf{C}}_{\tau \alpha \beta \gamma }=\varpi ^{2}\mathbf{C}_{\tau
\alpha \beta \gamma }\mbox{ and }\widehat{\mathbf{C}}_{\ \alpha \beta \gamma
}^{\tau }=\mathbf{C}_{\ \alpha \beta \gamma }^{\tau }.
\end{equation*}
\end{proposition}

\begin{proof}
Such transforms can be verified by a N--adapted calculus with respect to
fixed N--elongated (\ref{dder}) and (\ref{ddif}). We note here that with
respect to a coordinate frame, for a metric (\ref{metr}) \ with coefficients
(\ref{ansatz}), a transform (\ref{conftrdm}) define a nonlinear transform of
metric. The property of rescalling holds only for the d--metric coefficients
with respect to fixed data $\mathbf{N}=\{N_{i}^{a}\}.$

$\square $
\end{proof}

The Bianchi identities for $\mathbf{D}$,%
\begin{equation}
\mathbf{D}_{[\alpha }\mathbf{R}_{\ \tau \alpha ]\beta \gamma }=0,\mbox{
or }\mathbf{D}^{\tau }\ \widehat{\mathbf{C}}_{\gamma \tau \alpha \beta }=-2%
\mathbf{D}_{[\beta }\mathbf{P}_{\gamma ]\alpha },  \label{bianchid}
\end{equation}%
are standard ones with possible $h$- and $v$--projections \cite{vsgg}.

\begin{theorem}
For any fixed data $\left( \mathbf{g}_{\alpha \beta },N_{i}^{a}\right) ,$
there is a nonholonomic deformation to some $\left( \mathbf{g}_{\alpha
^{\prime }\beta ^{\prime }},N_{i^{\prime }}^{a^{\prime }}\right) $ for which
$\mathbf{C}_{\tau ^{\prime }\alpha ^{\prime }\beta ^{\prime }\gamma ^{\prime
}}=0$ with respect to a re-defined $\mathbf{e}^{\beta ^{\prime
}}=(e^{j^{\prime }},\mathbf{e}^{b^{\prime }})$ (\ref{ddif}).
\end{theorem}

\begin{proof}
Let us fix a d--metric (\ref{m1}) with coefficients $\mathbf{g}_{\alpha
^{\prime }\beta ^{\prime }}:=\varpi ^{2}(u)\mathbf{\eta }_{\alpha \beta }$
with $\mathbf{\eta }_{\alpha \beta }$ being diagonal constants of any
necessary signature $(\pm 1,\pm ,...,\pm ),$ with respect to some $\mathbf{e}%
^{\beta ^{\prime }}=(e^{j^{\prime }}=dx^{j^{\prime }},\mathbf{e}^{b^{\prime
}}=dy^{b^{\prime }}+N_{i^{\prime }}^{b^{\prime }}dx^{i^{\prime }}).$ For
such a d--metric and N--adapted co--bases, we can verify that $\mathbf{C}%
_{\tau ^{\prime }\alpha ^{\prime }\beta ^{\prime }\gamma ^{\prime }}=0,$ as
a consequence of Proposition \ref{prop1}. We can redefine data ( (\ref{m1}),$%
\mathbf{g}_{\alpha ^{\prime }\beta ^{\prime }})$ in a coordinate form (\ref%
{metr}) \ with coefficients (\ref{ansatz}) (with primed indices, $g_{%
\underline{\alpha }^{\prime }\underline{\beta }^{\prime }}$). Then
considering arbitrary frame transforms $e_{\ \underline{\alpha }}^{%
\underline{\alpha }^{\prime }}$ we can compute $g_{\underline{\alpha }%
\underline{\beta }}=e_{\ \underline{\alpha }}^{\underline{\alpha }^{\prime
}}e_{\ \underline{\beta }}^{\underline{\beta }^{\prime }}g_{\underline{%
\alpha }^{\prime }\underline{\beta }^{\prime }}.$ Finally, we can re--define
for a nonholonomic $2+2,$ or $4+4,$ splitting certain data $\left( \mathbf{g}%
_{\alpha \beta },N_{i}^{a}\right) ,$ for which, in general, $\mathbf{C}%
_{\tau \alpha \beta \gamma }\neq 0,$ and the corresponding to $\nabla
,C_{\tau \alpha \beta \gamma }\neq 0.$ Such construction with nonholonomic
deformations are possible because vierbeins (\ref{vbt}) may depend on some $%
N $--coefficients which can be present also in the generic off--diagonal
form of "primary" metric. The transformation laws of d--objects on
nonholonomic manifolds with N--connection are different from those on usual
manifolds without N--connection splitting (\ref{whitney}).

$\square $

\begin{conclusion}
\label{concla}An arbitrary (pseudo) Riemannian spacetime $V$ with metric
structure $\mathbf{g}=\{g_{\alpha \beta }\}$ is not conformally flat, i.e. $%
C_{\tau \alpha \beta \gamma }\neq 0,$ for $\nabla .$ Nevertheless, we can
always associate a nonholonomic manifold $\mathbf{V}$ enabled with the same
metric structure but with such a N--connection $\mathbf{N}$ when the
corresponding canonical d--connection $\mathbf{D}$ is with zero Weyl
d--tensor $\mathbf{C}_{\tau \alpha \beta \gamma }=0$ (we omit priming of
indices).
\end{conclusion}

The above values $C_{\tau \alpha \beta \gamma }$ and $\mathbf{C}_{\tau
\alpha \beta \gamma }$ are related in unique form by distortions (\ref%
{distrel}) for a unique $\mathbf{D}=\nabla +\mathbf{Q}$ (\ref{distrel}).
Such constructions depend on prescribed distribution $\mathbf{N.}$ They do
not violate a principle of general covariance on $V$, or $\mathbf{V.}$ We
can prescribe a necessary type distribution, adapted all constructions to
N--splitting, and then re--define everything in arbitrary systems of
reference.
\end{proof}

Finally, we note that similar statements can be formulated, up to some frame
transforms (\ref{ftransf}), for the cases when $(\mathbf{g,N,D})\rightarrow (%
\mathbf{\tilde{g},\tilde{N},\tilde{D}}),$ i.e. for a Finsler--Cartan space,
or any such variables on a (pseudo) Riemannian manifold.

\section{Finsler--Cartan Spinors and Einstein gravity}

\label{snfs}

Spinor and twistor geometries for data $\left( \mathbf{g},\mathbf{N},\mathbf{%
D}\right) $ and N--adapted frames can be elaborated \cite{vspinor,vhep,vstavr,vfncr,vsgg}
similarly to those for $(\mathbf{g},\nabla )$ in arbitrary frames of
reference \cite{penrr,manin,ward,huggett}.  The
concept of distinguished spinor, d--spinor, was introduced as a couple of $h$%
-- and $v$--spinors derived for a N--connection splitting (\ref{whitney}%
), see a brief summary in sections 2.2 and 3.1 of \cite{vfncr}.

\subsection{Spinors and N--connections}

We provide main definitions and introduce an abstract index formalism
adapted nonholonomic manifolds/bundles with $n+m$ splitting

\subsubsection{Clifford N--adapted structures and spin d--connections}

\begin{definition}
We define a Clifford d--algebra as a\ $ \wedge V^{n+m}$ algebra determined by a
product $\mathbf{u}\mathbf{v}+\mathbf{v}\mathbf{u}=2\mathbf{g}(\mathbf{u},%
\mathbf{v})\ \mathbb{I}$, with associated h--, v--products{\small
\begin{equation*}
\ ^{h}u~^{h}v+~^{h}v~^{h}u=2~^{h}g(u,v)\ ^{h}\mathbb{I},\ ^{v}u\ \
~^{v}v+~^{v}v\ \ ~^{v}u=2\ ^{v}h(~^{v}u,\ \ ~^{v}v) \ ^{v}\mathbb{I},
\end{equation*}%
} for any $\mathbf{u}=(~^{h}u,\ ~^{v}u),\ \mathbf{v}=(~^{h}v,\ ~^{v}v)\in
V^{n+m},$ where $\mathbb{I},$ $\ ~^{h}\mathbb{I}\ $\ and $^{v}\mathbb{I}\ \ $%
\ are unity matrices of corresponding dimensions $(n+m)\times (n+m),$ or $%
n\times n$ and $m\times m.$ \footnote{%
in certain cases, we shall consider only \textquotedblright
horizontal\textquotedblright\ geometric constructions if they are similar to
\textquotedblright vertical\textquotedblright\ ones}
\end{definition}

Any metric $^{h}g$ on $h\mathbf{V}$ is defined by sections of $T~h\mathbf{V}
$ provided with a bilinear symmetric form on continuous sections $Sec(T~h%
\mathbf{V}).$ We can define Clifford h--algebras $~^{h}\mathcal{C}l(T_{x}h%
\mathbf{V}),$ $\gamma _{i}\gamma _{j}+\gamma _{j}\gamma _{i}=2\ g_{ij}~^{h}%
\mathbb{I},\ $\ in any point $x\in T~h\mathbf{V}.$

The Clifford d--module of a vector bundle ${E}$ (in general, we can consider
a complex vector bundle $~^{E}\pi :\ E\rightarrow \mathbf{V}$) is defined by
the $C(\mathbf{V})$--module $Sec({E})$ of continuous sections in ${E},$ $c:\
Sec(~^{N}\mathcal{C}l(\mathbf{V}))\rightarrow End(Sec({E}))$. Prescribing a
N--connection structure, a Clifford N--anholonomic bundle on $\mathbf{V}$ is
by definition $~^{N}\mathcal{C}l(\mathbf{V})\doteq ~^{N}\mathcal{C}l(T^{\ast
}\mathbf{V}),$ where $T^{\ast }$ is the dual tangent bundle.

\begin{definition}
A Clifford d--space associated to data $\ \mathbf{g}(x,y)$ (\ref{m1}) and $%
\mathbf{N}$ for a nonholonomic manifold $~\mathbf{V}$ is defined as a
Clifford bundle
\begin{equation*}
~\mathcal{C}l(\mathbf{V})=~^{h}\mathcal{C}l(h\mathbf{V})\oplus ~^{v}\mathcal{%
C}l(v\mathbf{V}),
\end{equation*}
with Clifford h--space, $~^{h}\mathcal{C}l(h\mathbf{V})\doteq ~^{h}\mathcal{C%
}l(T^{\ast }h\mathbf{V}),$ and Clifford v--space,\newline
$^{v}\mathcal{C}l(v\mathbf{V})\doteq ~^{v}\mathcal{C}l(T^{\ast }v\mathbf{V}%
). $
\end{definition}

Let $V^{n}$ be a vector space provided with Clifford structure. We write $%
^{h}V^{n}$ if \ its tangent space is provided with a quadratic form $\ ^{h}g$
and consider $~^{h}\mathcal{C}l(V^{n})\equiv \mathcal{C}l(~^{h}V^{n})$ using
the subgroup $SO(\ ^{h}V^{n})\subset O(\ ^{h}V^{n}).$ A standard definition
of spinors is possible using sections of a vector bundle $S$ on a manifold $%
M $ being considered an irreducible representation of the group $%
Spin(M)\doteq Spin(T_{x}^{\ast }M)$ defined on the typical fiber. The set of
sections $Sec(S)$ defines an irreducible Clifford module.

The space of complex h--spins is defined by the subgroup $\ $%
\begin{equation*}
^{h}Spin^{c}(n)\equiv Spin^{c}(\ ^{h}V^{n})\equiv \
^{h}Spin^{c}(V^{n})\subset \mathcal{C}l(\ ^{h}V^{n}),
\end{equation*}%
determined by the products of pairs of vectors $w\in \ ^{h}V^{\mathbb{C}}$
when $w\doteq pu$ where $p$ is a complex number of module 1 and $u$ is of
unity length in $\ ^{h}V^{n}.$ Similar constructions can be performed for
the v--subspace $~^{v}V^{m},$ which allows us to define similarly the group
of real v--spins. A h--spinor bundle $\ ^{h}S$ on a h--space $h\mathbf{V}$
is a complex vector bundle with both defined action of the h--spin group $\
^{h}Spin(V^{n})$ on the typical fiber and an irreducible representation of
the group $\ ^{h}Spin(\mathbf{V})\equiv Spin(h\mathbf{V})\doteq
Spin(T_{x}^{\ast }h\mathbf{V}).$ The set of sections $Sec(\ ^{h}S)$ defines
an irreducible Clifford h--module.

\begin{definition}
\label{ddsp} A distinguished spinor (d--spinor) bundle $\mathbf{S}\doteq (\
\ ^{h}S,\ \ \ ^{v}S)$ for $\mathbf{V},$ $\ dim\mathbf{V}=n+m,$ is a complex
vector bundle with an action of the spin distinguished (d--group) $Spin\
\mathbf{V}\doteq Spin(V^{n})\oplus Spin(V^{m})$ with an irreducible
representation $Spin(\mathbf{V})\doteq Spin(T^{\ast }\mathbf{V}).$ The set
of sections $Sec(\mathbf{S})=Sec(\ \ ^{h}{S})\oplus Sec(\ \ \ ^{v}{S})$ is
an irreducible Clifford d--module.
\end{definition}

The considerations presented above provide a proof for\

\begin{theorem}
\label{mr1}Any d--metric and N--con\-nec\-ti\-on structures define naturally
the fundamental geometric objects and structures (such as the Clifford
h--module, v--module and Clifford d--modules,or the h--spin, v--spin
structures and d--spinors) for the corresponding nonholonomic spin manifold
and/or N--anholo\-nom\-ic spinor (d--spinor) manifold.
\end{theorem}

We consider a Hilbert space of finite dimension and denote a local dual
coordinate basis $e^{\underline{i}}\doteq dx^{\underline{i}}$ on $h\mathbf{V.%
}$\ In N--adapted form, it is possible to introduce certain classes of
orthonormalized vielbeins and the N--adapted vielbeins, $e^{\hat{\imath}%
}\doteq e_{\ \underline{i}}^{\hat{\imath}}(x,y)\ e^{\underline{i}}$ and $%
e^{i}\doteq e_{\ \underline{i}}^{i}(x,y)\ e^{\underline{i}},$ when $g^{%
\underline{i}\underline{j}}\ e_{\ \underline{i}}^{\hat{\imath}}e_{\
\underline{j}}^{\hat{\jmath}}=\delta ^{\hat{\imath}\hat{\jmath}}$ and $g^{%
\underline{i}\underline{j}}\ e_{\ \underline{i}}^{i}e_{\ \underline{j}%
}^{j}=g^{ij}.$ This allows us to define the algebra of Dirac's gamma
h--matrices with self--adjoint matrices $M_{k}(\mathbb{C}),$ where $%
k=2^{n/2} $ is the dimension of the irreducible representation of $\mathcal{C%
}l(h\mathbf{V})$ \ derived from the relation $\gamma ^{\hat{\imath}}\gamma ^{%
\hat{\jmath}}+\gamma ^{\hat{\jmath}}\gamma ^{\hat{\imath}}=2\delta ^{\hat{%
\imath}\hat{\jmath}}\ ^{h}\mathbb{I}.$ The action of $dx^{i}\in \mathcal{C}%
l(h\mathbf{V})$ on a spinor $\ ^{h}\psi \in \ ^{h}S$ can be parameterized by
formulas $\ ^{h}c(dx^{\hat{\imath}})\doteq \gamma ^{\hat{\imath}}$ and $\
^{h}c(dx^{i})\ \ ^{h}\psi \doteq \gamma ^{i}\ \ ^{h}\psi \equiv e_{\ \hat{%
\imath}}^{i}\ \gamma ^{\hat{\imath}}\ \ ^{h}\psi$. The algebra of Dirac's
gamma v--matrices is defined by self--adjoint matrices $M_{k}^{\prime }(%
\mathbb{C}),$ where $k^{\prime }=2^{m/2}$ is the dimension of the
irreducible representation of $\mathcal{C}l(F),$ for a typical fiber $F),$
when $\gamma ^{\hat{a}}\gamma ^{\hat{b}}+\gamma ^{\hat{b}}\gamma ^{\hat{a}%
}=2\delta ^{\hat{a}\hat{b}}\ ^{v}\mathbb{I}.$ The action of $dy^{a}\in
\mathcal{C}l(F)$ on a spinor $\ ^{v}\psi \in \ ^{v}S$ is $\ ^{v}c(dy^{\hat{a}%
})\doteq \gamma ^{\hat{a}}$ and $\ ^{v}c(dy^{a})\ ^{v}\psi \doteq \gamma
^{a}\ ^{v}\psi \equiv e_{\ \hat{a}}^{a}\ \gamma ^{\hat{a}}\ ^{v}\psi$.

In general, a matrix calculus with gamma d--matrices can be elaborated for a
total d--metric structure $\mathbf{g}=~^{h}g\oplus ~^{v}h.$ We consider
d--spinors $\breve{\psi}\doteq (~^{h}\psi ,\ ~^{v}\psi )\in \mathbf{S}\doteq
(~^{h}S,\ ~^{v}S)$ and d--gamma matrix relations $\ \gamma ^{\hat{\alpha}%
}\gamma ^{\hat{\beta}}+\gamma ^{\hat{\beta}}\gamma ^{\hat{\alpha}}=2\delta ^{%
\hat{\alpha}\hat{\beta}}\mathbb{\ I}.$ The action of $du^{\alpha }\in
\mathcal{C}l(\mathbf{V})$ on a d--spinor $\breve{\psi}\in \ \mathbf{S}$
resulting in distinguished irreducible representations $\mathbf{c}(du^{\hat{%
\alpha}})\doteq \gamma ^{\hat{\alpha}}$ and
\begin{equation*}
\mathbf{c}=(du^{\alpha })\ \breve{\psi}\doteq \gamma ^{\alpha }\ \breve{\psi}%
\equiv e_{\ \hat{\alpha}}^{\alpha }\ \gamma ^{\hat{\alpha}}\ \breve{\psi}.
\end{equation*}%
We obtain d--metric -- d--gamma matrix relations
\begin{equation*}
\gamma ^{\alpha }(u)\gamma ^{\beta }(u)+\gamma ^{\beta }(u)\gamma ^{\alpha
}(u)=2g^{\alpha \beta }(u)\ \mathbb{I},
\end{equation*}%
which can re--written for "boldface" coefficients of metric. In irreducible
form $\breve{\gamma}\doteq ~^{h}\gamma \oplus \ ~^{v}\gamma $ and $\breve{%
\psi}\doteq ~^{h}\psi \oplus ~^{v}\psi ,$ or, $\gamma ^{\alpha }\doteq
(~^{h}\gamma ^{i},~^{v}\gamma ^{a})$ and $\breve{\psi}\doteq (~^{h}\psi ,\
~^{v}\psi ).$

The spin connection $~_{S}\nabla $ for (pseudo) Riemannian manifolds is
standardly determined by the Levi--Civita connection, $~_{S}\nabla \doteq d-%
\frac{1}{4}\ \Gamma _{\ jk}^{i}\gamma _{i}\gamma ^{j}\ dx^{k}.$ Similar
constructions are possible for nonholonomic manifolds enabled with metric
compatible d--connections (for instance, in Finsler--Cartan geometry). The
spin d--connection operators $~_{\mathbf{S}}\mathbf{\nabla }$ can be
similarly constructed from any metric compatible d--connection ${\mathbf{%
\Gamma }}_{\ \beta \mu }^{\alpha }$ (for instance, with coefficients (\ref%
{candcon}), or (\ref{cdc})), when for a scalar function $f(x,y)$ in the form
\begin{equation*}
\delta f=\left( \mathbf{e}_{\nu }f\right) \delta u^{\nu }=\left( \mathbf{e}%
_{i}f\right) dx^{i}+\left( e_{a}f\right) \delta y^{a},
\end{equation*}%
for $\delta u^{\nu }=\mathbf{e}^{\nu }$ (\ref{ddif}).

\begin{definition}
\label{spdcon}The canonical (Finsler--Cartan) spin d--connection is defined
by $\mathbf{D=\{\mathbf{\Gamma }_{\ \beta \mu }^{\alpha }\}}$ ($\mathbf{%
\tilde{D}=\{\mathbf{\tilde{\Gamma}}_{\ \beta \mu }^{\alpha }\}}$) following
formula
\begin{equation}
~_{\mathbf{S}}\mathbf{D}\doteq \delta -\frac{1}{4}\ \mathbf{\Gamma }_{\
\beta \mu }^{\alpha }\gamma _{\alpha }\gamma ^{\beta }\mathbf{e}^{\mu }\ (~_{%
\mathbf{S}}\mathbf{\tilde{D}}\doteq \delta -\frac{1}{4}\ \mathbf{\tilde{%
\Gamma}}_{\ \beta \mu }^{\alpha }\gamma _{\alpha }\gamma ^{\beta }\mathbf{e}%
^{\mu }).  \notag
\end{equation}
\end{definition}

For the purposes of this work, we shall consider abstract index formulations
of d--spinor calculus for nonholonomic manifolds/bundles with splitting $2+2,
$ or $4+4.$

\subsubsection{Abstract d--tensor and d--spinor indices}

Indices of d--tensors are considered as a set of labels which can changed
into respective sets of d--spinor indices, primed and unprimed (with dots
and without dots) following, for instance, such rules: $\Psi _{\quad \mu
}^{\alpha \beta }=\Psi _{\qquad \quad \dot{M}\dot{M}^{\prime }}^{\dot{A}\dot{%
A}^{\prime }\dot{B}\dot{B}^{\prime }}$, where dot spinor capital indices
correspond small Greek tensor indices. For $h$- and/or $v$--decompositions,
when $\xi ^{\alpha }=(\xi ^{i},\xi ^{a}),$ we shall write $\xi ^{\dot{A}\dot{%
A}^{\prime }}=(\xi ^{II^{\prime }},\xi ^{AA^{\prime }}),$ where $\xi
^{II^{\prime }}$ is for a horizontal spinor--vector and $\xi ^{AA^{\prime }}$
is for vertical spinor--vector. In similar forms, we can consider $h$- and $%
v $- and spinor decompositions for forms and tensors with mixed indices. So,
we shall follow the formalism from \cite{penrr} but re--defined in a form to
be able to encode spinorially d--tensors with possible N--adapted splitting.
 Primed spinor indices are complex conjugated with
corresponding unprimed, for instance, $\overline{\xi ^{\dot{A}\dot{A}%
^{\prime }}}=\overline{\xi }^{\dot{A}^{\prime }\dot{A}},\overline{\xi
^{II^{\prime }}}=\overline{\xi }^{I^{\prime }I}$ etc both for up and low
indices.

We can consider antisymmetric $\varepsilon $--spinors on total spaces and $h$%
- and $v$-subspaces with the properties,%
\begin{eqnarray*}
\varepsilon ^{\dot{A}^{\prime }\dot{B}^{\prime }} &:&=\overline{\varepsilon }%
^{\dot{A}^{\prime }\dot{B}^{\prime }}=\overline{\varepsilon ^{\dot{A}\dot{B}}%
},\varepsilon _{\dot{A}^{\prime }\dot{B}^{\prime }}:=\overline{\varepsilon }%
_{\dot{A}^{\prime }\dot{B}^{\prime }}=\overline{\varepsilon _{\dot{A}\dot{B}}%
}; \\
\varepsilon ^{I^{\prime }J^{\prime }} &:&=\overline{\varepsilon }^{I^{\prime
}J^{\prime }}=\overline{\varepsilon ^{IJ}},\varepsilon _{I^{\prime
}J^{\prime }}:=\overline{\varepsilon }_{I^{\prime }J^{\prime }}=\overline{%
\varepsilon _{IJ}}\mbox{ \ for h--spinor indices}; \\
\varepsilon ^{A^{\prime }B^{\prime }} &:&=\overline{\varepsilon }^{A^{\prime
}B^{\prime }}=\overline{\varepsilon ^{AB}},\varepsilon _{A^{\prime
}B^{\prime }}:=\overline{\varepsilon }_{A^{\prime }B^{\prime }}=\overline{%
\varepsilon _{AB}}\mbox{ \ for v--spinor indices}.
\end{eqnarray*}%
This is related to the rules of transforming low indices into up ones, and
inversely, using metrics and/or d--metrics, for instance, $\mathbf{g}%
_{\alpha \beta }=[g_{ij},g_{ab}]$ and $\mathbf{g}^{\alpha \beta
}=[g^{ij},g^{ab}].$ In brief, the spinor decompositions of metrics are
written in the form,%
\begin{eqnarray}
\mathbf{g}_{\alpha \beta } &=&\varepsilon _{\dot{A}\dot{B}}\varepsilon _{%
\dot{A}^{\prime }\dot{B}^{\prime }},\mathbf{g}^{\alpha \beta }=\varepsilon ^{%
\dot{A}\dot{B}}\varepsilon ^{\dot{A}^{\prime }\dot{B}^{\prime }};
\label{dmetrspindec} \\
g_{ij} &=&\varepsilon _{IJ}\varepsilon _{I^{\prime }J^{\prime
}},g^{ij}=\varepsilon ^{IJ}\varepsilon ^{I^{\prime }J^{\prime }},%
\mbox{ \ for
h--metrics};  \notag \\
g_{ab} &=&\varepsilon _{AB}\varepsilon _{A^{\prime }B^{\prime
}},g^{ab}=\varepsilon ^{AB}\varepsilon ^{A^{\prime }B^{\prime }},%
\mbox{ \ for
h--metrics}.  \notag
\end{eqnarray}%
In our works \cite{vspinor,vhep,vstavr,vfncr,vsgg}, we used also N--adapted
gamma matrices generating corresponding Clifford algebras for spinors (some
authors call them $\sigma $--symbols or transition indices from Minkowski
tetrads to spin systems of reference, on corresponding tangent bundles). In
brief, such a formalism is related to orthonormalized (co) bases, $\mathbf{e}%
_{\alpha ^{\prime }}=(\mathbf{e}_{i^{\prime }},e_{a^{\prime }})$ and $%
\mathbf{e}_{\ }^{\beta ^{\prime }}=(e^{j^{\prime }},\mathbf{e}_{\
}^{b^{\prime }}),$ where d--tensor primed indices are used for definition of
$4\times 4$ $\gamma $--matrices $\gamma _{\alpha ^{\prime }}=(\gamma
_{\alpha ^{\prime }}^{\dot{A}\dot{A}^{\prime }})$ \ satisfying the relations
\begin{equation}
\gamma _{\alpha ^{\prime }}\gamma _{\beta ^{\prime }}+\gamma _{\beta
^{\prime }}\gamma _{\alpha ^{\prime }}=2\eta _{\alpha ^{\prime }\beta
^{\prime }},  \label{gammamatr}
\end{equation}%
where the Minkowski metric $\eta _{\alpha ^{\prime }\beta ^{\prime }}$ is,
for instance, of signature $(+++-)$ for a formal $2+2$ splitting\footnote{%
on tangent bundles to Lorentz manifolds of dimension 8=4+4, we can use one
such a gamma relation for the h--subspace and another one for the v--subspace%
}. Using transforms of type $\mathbf{e}_{\alpha ^{\prime }}^{\ \underline{%
\alpha }}=\mathbf{e}_{\alpha ^{\prime }}^{\ \alpha }\mathbf{e}_{\alpha }^{\
\underline{\alpha }}$ and their inverse, we can write
\begin{equation*}
\gamma _{\alpha }\gamma _{\beta }+\gamma _{\beta }\gamma _{\alpha }=2\mathbf{%
g}_{\alpha \beta }\mbox{ and }\gamma _{\underline{\alpha }}\gamma _{%
\underline{\beta }}+\gamma _{\underline{\beta }}\gamma _{\underline{\alpha }%
}=2\mathbf{g}_{\underline{\alpha }\underline{\beta }}
\end{equation*}%
where d--metric (\ref{m1}) and, respectively, (\ref{metr}) and (\ref{ansatz}%
), including N--coefficients, are considered for $\mathbf{e}_{\alpha }=%
\mathbf{e}_{\alpha }^{\ \underline{\alpha }}\partial _{\underline{\alpha }}$
and $\mathbf{e}_{\ }^{\beta }=\mathbf{e}_{\ \underline{\beta }}^{\beta }du^{%
\underline{\beta }}$ with decompositions of type (\ref{nadtetr}) and (\ref%
{vbt}). With $\gamma $--matrices, for instance the first relation in (\ref%
{dmetrspindec}) is written
\begin{equation}
\mathbf{g}_{\alpha \beta }=\gamma _{\alpha }^{\dot{A}\dot{A}^{\prime
}}\gamma _{\beta }^{\dot{B}\dot{B}^{\prime }}\varepsilon _{\dot{A}\dot{B}%
}\varepsilon _{\dot{A}^{\prime }\dot{B}^{\prime }}  \label{gammametrd}
\end{equation}%
for $\gamma _{\alpha ^{\prime }}^{\dot{A}\dot{A}^{\prime }}:=\mathbf{e}%
_{\alpha ^{\prime }}^{\ \alpha }\gamma _{\alpha }^{\dot{A}\dot{A}^{\prime }}$
(we omit similar decompositions for $h$- and $v$--indices). For simplicity,
in this work we shall follow abstract algebraic decompositions not writing
gamma matrices even formulas of type (\ref{gammametrd}) are necessary for
constructing in explicit form exact generic off--diagonal solutions with
nontrivial N--coefficients of Einstein--Dirac systems.

\subsection{N--adapted spinors and nonholonomic (Finsler) gravity}

\subsubsection{N--adapted covariant derivatives and spin coefficients}

In brief, we shall write the d--spinor equivalents as
\begin{equation*}
\mathbf{D}_{\alpha }=\mathbf{D}_{\dot{A}\dot{A}^{\prime }}=\mathbf{D}_{\dot{A%
}^{\prime }\dot{A}},D_{i}=D_{II^{\prime }}=D_{I^{\prime
}I},D_{a}=D_{AA^{\prime }}=D_{A^{\prime }A}
\end{equation*}%
etc. Fixing spin diads $\varepsilon _{\dot{A}}^{\ \underline{\dot{A}}%
}=(\varepsilon _{\dot{A}}^{\ 0},\varepsilon _{\dot{A}}^{\ 1}),\varepsilon
_{I}^{\ \underline{I}}=(\varepsilon _{I}^{\ 0},\varepsilon _{I}^{\
1}),\varepsilon _{A}^{\ \underline{A}}=(\varepsilon _{A}^{\ 0},\varepsilon
_{A}^{\ 1})$, and theirs respective duals, $\varepsilon _{\underline{\dot{A}}%
}^{\ \dot{A}}, $ $\varepsilon _{\underline{I}}^{\ I},\varepsilon _{%
\underline{A}}^{\ A},$\footnote{%
in our approach the underlined indices are equivalent to "boldface" indices
in \cite{penrr}} we can introduce N--adapted d--spin coefficients, {\small
\begin{eqnarray*}
&&\gamma _{\underline{\dot{A}}\underline{\dot{A}}^{\prime }\underline{\dot{C}%
}}^{\qquad \underline{\dot{B}}} :=\varepsilon _{\dot{A}}^{\ \underline{\dot{B%
}}}\mathbf{D}_{\underline{\dot{A}}\underline{\dot{A}}^{\prime }}\varepsilon
_{\underline{\dot{C}}}^{\ \dot{A}}=-\varepsilon _{\underline{\dot{C}}}^{\
\dot{A}}\mathbf{D}_{\underline{\dot{A}}\underline{\dot{A}}^{\prime
}}\varepsilon _{\dot{A}}^{\ \underline{\dot{B}}}, \\
&&\gamma _{\underline{I}\underline{I}^{\prime }\underline{K}}^{\qquad
\underline{J}} :=\varepsilon _{I}^{\ \underline{J}}\mathbf{D}_{\underline{I}%
\underline{I}^{\prime }}\varepsilon _{\underline{K}}^{\ I}=-\varepsilon _{%
\underline{I}}^{\ I}\mathbf{D}_{\underline{I}\underline{I}^{\prime
}}\varepsilon _{I}^{\ \underline{J}}, \gamma _{\underline{A}\underline{A}%
^{\prime }\underline{C}}^{\qquad \underline{B}} :=\varepsilon _{A}^{\
\underline{B}}\mathbf{D}_{\underline{A}\underline{A}^{\prime }}\varepsilon _{%
\underline{C}}^{\ A}=-\varepsilon _{\underline{A}}^{\ A}\mathbf{D}_{%
\underline{A}\underline{A}^{\prime }}\varepsilon _{A}^{\ \underline{B}},
\end{eqnarray*}%
} which are equivalent to the spin d--connection from Definition \ref{spdcon}%
. This way, we can introduce a canonical and/or Cartan type null--tetradic
type calculus etc. For simplicity, we shall omit in the future spinor $h$-
and $v$--index formulas if that will not result in ambiguities or lost of
some important properties.

\subsubsection{Spinor d--curvature and Bianchi identities}

Following a N--adapted d--spinor calculus with abstract indices for
nonholonomic splitting $2+2$ (for simplicity, we shall consider "dot" spinor
indices; only some examples for Finsler--Cartan configurations will be considered), we prove:

\begin{theorem}
\label{thb1}In d--spinor variables,

\begin{itemize}
\item the canonical d--commutator (\ref{dcommut}) (d--torsion (\ref{dtors}))
is
\begin{eqnarray*}
\mathbf{\Delta }_{\alpha \beta } &=&\varepsilon _{\dot{A}\dot{B}}\mathbf{%
\square }_{\dot{A}^{\prime }\dot{B}^{\prime }}+\varepsilon _{\dot{A}^{\prime
}\dot{B}^{\prime }}\square _{\dot{A}\dot{B}}, \\
\mathbf{\Delta }_{ij} &=&\varepsilon _{IJ}\square _{I^{\prime }J^{\prime
}}+\varepsilon _{I^{\prime }J^{\prime }}\square _{IJ},,\mathbf{\Delta }%
_{ia}=\varepsilon _{IA}\square _{I^{\prime }A^{\prime }}+\varepsilon
_{I^{\prime }A^{\prime }}\square _{IA},...
\end{eqnarray*}%
where $\square _{\dot{A}\dot{B}}:=\mathbf{D}_{\dot{A}^{\prime }(\dot{A}}%
\mathbf{D}_{\dot{B})}^{\dot{A}^{\prime }}$ and $\mathbf{\square }_{\dot{A}%
^{\prime }\dot{B}^{\prime }}:=\mathbf{D}_{\dot{A}(\dot{A}^{\prime }}\mathbf{D%
}_{\dot{B}^{\prime })}^{\dot{A}}$ etc;

\item the spinor d--commutators acting on a d--spinor $\varkappa _{\dot{C}}$
result in%
\begin{eqnarray*}
\square _{\dot{A}\dot{B}}\varkappa _{\dot{C}} &=&[\Psi _{\dot{C}\dot{T}\dot{A%
}\dot{B}}+\Lambda (\varepsilon _{\dot{A}\dot{C}}\varepsilon _{\dot{B}\dot{T}%
}+\varepsilon _{\dot{A}\dot{T}}\varepsilon _{\dot{B}\dot{C}})]\varkappa ^{%
\dot{T}}, \\
\square _{\dot{A}^{\prime }\dot{B}^{\prime }}\varkappa _{\dot{C}} &=&\Phi _{%
\dot{C}\dot{T}\dot{A}^{\prime }\dot{B}^{\prime }}\varkappa ^{\dot{T}};
\end{eqnarray*}

\item the Riemann d--tensor (\ref{dcurv}) is {\small
\begin{eqnarray*}
&&\mathbf{R}_{\ \tau \gamma \alpha \beta }=\Psi _{\dot{T}\dot{C}\dot{A}\dot{B%
}}\varepsilon _{\dot{A}^{\prime }\dot{B}^{\prime }}\varepsilon _{\dot{T}%
^{\prime }\dot{C}^{\prime }}+\overline{\Psi }_{\dot{T}^{\prime }\dot{C}%
^{\prime }\dot{A}^{\prime }\dot{B}^{\prime }}\varepsilon _{\dot{A}\dot{B}%
}\varepsilon _{\dot{T}\dot{C}}+\Phi _{\dot{T}^{\prime }\dot{C}^{\prime }\dot{%
A}\dot{B}}\varepsilon _{\dot{A}^{\prime }\dot{B}^{\prime }}\varepsilon _{%
\dot{T}\dot{C}} \\
&&+\overline{\Phi }_{\dot{T}\dot{C}\dot{A}^{\prime }\dot{B}^{\prime
}}\varepsilon _{\dot{A}\dot{B}}\varepsilon _{\dot{T}^{\prime }\dot{C}%
^{\prime }}+2\Lambda (\varepsilon _{\dot{A}\dot{C}}\varepsilon _{\dot{B}\dot{%
T}}\varepsilon _{\dot{A}^{\prime }\dot{C}^{\prime }}\varepsilon _{\dot{B}%
^{\prime }\dot{T}^{\prime }}-\varepsilon _{\dot{A}\dot{T}}\varepsilon _{\dot{%
B}\dot{C}}\varepsilon _{\dot{A}^{\prime }\dot{T}^{\prime }}\varepsilon _{%
\dot{B}^{\prime }\dot{C}^{\prime }}),
\end{eqnarray*}%
} for $\Lambda =\overline{\Lambda }=\frac{1}{24}\ _{s}R;$

\item the Weyl conformal d--tensor (\ref{wtcd}) splits into anti--selfdual,%
\begin{equation*}
\ ^{-}\mathbf{C}_{\ \alpha \beta \gamma \tau }:=\Psi _{\dot{A}\dot{B}\dot{C}%
\dot{T}}\varepsilon _{\dot{A}^{\prime }\dot{B}^{\prime }}\varepsilon _{\dot{C%
}^{\prime }\dot{T}^{\prime }},
\end{equation*}%
and selfdual parts,%
\begin{equation*}
\ ^{+}\mathbf{C}_{\ \alpha \beta \gamma \tau }:=\overline{\Psi }_{\dot{A}%
^{\prime }\dot{B}^{\prime }\dot{C}^{\prime }\dot{T}^{\prime }}\varepsilon _{%
\dot{A}\dot{B}}\varepsilon _{\dot{C}\dot{T}},
\end{equation*}%
when $\mathbf{C}_{\ \alpha \beta \gamma \tau }=\ ^{-}\mathbf{C}_{\ \alpha
\beta \gamma \tau }+\ ^{+}\mathbf{C}_{\ \alpha \beta \gamma \tau }$.
\end{itemize}
\end{theorem}

Similar N--adapted $2+2$ decompositions can be computed for $\tilde{D}_{i}=%
\tilde{D}_{II^{\prime }}=\tilde{D}_{I^{\prime }I},\tilde{D}_{a}=\tilde{D}%
_{AA^{\prime }}=\tilde{D}_{A^{\prime }A}.$ For instance, the components of the
Finsler--Cartan curvature (\ref{curvcart}) can be written {\small
\begin{eqnarray}
\tilde{R}_{\ ijkh} &=&\tilde{\Psi}_{IJKH}\varepsilon _{I^{\prime }J^{\prime
}}\varepsilon _{K^{\prime }H^{\prime }}+\overline{\tilde{\Psi}}_{I^{\prime
}J^{\prime }K^{\prime }H^{\prime }}\varepsilon _{IJ}\varepsilon _{KH}+\tilde{%
\Phi}_{IJK^{\prime }H^{\prime }}\varepsilon _{I^{\prime }J^{\prime
}}\varepsilon _{KH}+  \label{curvfcsp} \\
&&\overline{\tilde{\Phi}}_{I^{\prime }J^{\prime }KH}\varepsilon
_{IJ}\varepsilon _{K^{\prime }H^{\prime }}+2\Lambda (\varepsilon
_{IK}\varepsilon _{JH}\varepsilon _{I^{\prime }K^{\prime }}\varepsilon
_{J^{\prime }H^{\prime }}-\varepsilon _{IH}\varepsilon _{JK}\varepsilon
_{IH}\varepsilon _{JK}),  \notag \\
\tilde{P}_{\ ijka} &=&\tilde{\Psi}_{IJKA}\varepsilon _{I^{\prime }J^{\prime
}}\varepsilon _{K^{\prime }A^{\prime }}+\overline{\tilde{\Psi}}_{I^{\prime
}J^{\prime }K^{\prime }A^{\prime }}\varepsilon _{IJ}\varepsilon _{KA}+\tilde{%
\Phi}_{IJK^{\prime }A^{\prime }}\varepsilon _{I^{\prime }J^{\prime
}}\varepsilon _{KA}+  \notag \\
&&\overline{\tilde{\Phi}}_{I^{\prime }J^{\prime }KA}\varepsilon
_{IJ}\varepsilon _{K^{\prime }A^{\prime }}+2\Lambda (\varepsilon
_{IK}\varepsilon _{JA}\varepsilon _{I^{\prime }K^{\prime }}\varepsilon
_{J^{\prime }A^{\prime }}-\varepsilon _{IA}\varepsilon _{JK}\varepsilon
_{IA}\varepsilon _{JK}),  \notag \\
&&...  \notag
\end{eqnarray}%
} We note that we can apply the formalism from \cite{penrr} for any $h$- and
$v$--values (with not "dot" indices) and, in general form for small Greek
d--tensor indices, with "dot" indices), if we work in N--adapted frames and
keep in mind that Finsler like d--connections $\mathbf{D}_{\alpha },$ or $%
\mathbf{\tilde{D}}_{\alpha },$ are with torsions completely determined by
data $(\mathbf{g,N}).$ Following such rules, we obtain proofs for

\begin{theorem}
\label{thb2}

\begin{itemize}
\item The Bianchi identities (\ref{bianchid}) transform into
\begin{equation*}
\mathbf{D}_{\dot{B}^{\prime }}^{\dot{A}}\Psi _{\dot{A}\dot{B}\dot{C}\dot{T}}=%
\mathbf{D}_{(\dot{B}}^{\dot{A}^{\prime }}\Phi _{\dot{C}\dot{T})\dot{A}%
^{\prime }\dot{B}^{\prime }},\ \mathbf{D}^{\dot{C}\dot{A}^{\prime }}\Phi _{%
\dot{C}\dot{T}\dot{A}^{\prime }\dot{B}^{\prime }}=-3\mathbf{D}_{\dot{T}\dot{B%
}^{\prime }}\Lambda ;
\end{equation*}

\item the Einstein d--equations (\ref{deinsteq}), (\ref{cdeinst}) and/or (%
\ref{einst1}) transform the first identity into\newline $\mathbf{D}_{\dot{B}^{\prime
}}^{\dot{A}}\Psi _{\dot{A}\dot{B}\dot{C}\dot{T}}=8\pi G\mathbf{D}_{(\dot{B}%
}^{\dot{A}^{\prime }}\mathbf{T}_{\dot{C}\dot{T})\dot{A}^{\prime }\dot{B}%
^{\prime }}$, see also (\ref{einst2});

\item the vacuum field equations for locally anisotropic models with
nontrivial cosmological constant $\lambda =6\Lambda $ are
\begin{equation*}
\Phi _{\dot{C}\dot{T}\dot{A}^{\prime }\dot{B}^{\prime }}=0,\mathbf{D}_{\dot{A%
}\dot{B}^{\prime }}\Psi _{\dot{A}\dot{B}\dot{C}\dot{T}}=0;
\end{equation*}

\item in Cartan--Finsler gravity models with $2+2$ splitting, similar
formulas hold for $\mathbf{D\rightarrow \tilde{D};}$

\item all equations from this theorem transform into similar ones for the
Levi--Civita connection $\nabla $ if and only if $\mathbf{\Delta }_{\alpha
\beta }=0,$ or $\mathbf{\tilde{\Delta}}_{\alpha \beta }=0,$ see (\ref%
{dcommut}), which is equivalent to (\ref{lccond}).
\end{itemize}
\end{theorem}

\subsubsection{N--adapted conformal transforms}

Let us introduce the value
\begin{equation}
\Upsilon _{\alpha }:=\varpi ^{-1}\mathbf{D}_{\alpha }\varpi =\mathbf{D}%
_{\alpha }\ln \varpi ,  \label{confdvect}
\end{equation}
where the nonzero positive function $\varpi (u)$ is taken for conformal
transforms $\widehat{\mathbf{g}}_{\alpha \beta }:=\varpi ^{2}\mathbf{g}%
_{\alpha \beta }$ (\ref{conftrdm}) from Proposition \ref{prop1}. Using the
first formula in (\ref{dmetrspindec}), we conclude that in N--adapted form%
\begin{equation*}
\varepsilon _{\dot{A}\dot{B}}\rightarrow \widehat{\varepsilon }_{\dot{A}\dot{%
B}}=\varpi \varepsilon _{\dot{A}\dot{B}}\mbox{ and }\varepsilon _{\dot{A}%
^{\prime }\dot{B}^{\prime }}\rightarrow \widehat{\varepsilon }_{\dot{A}%
^{\prime }\dot{B}^{\prime }}=\varpi \varepsilon _{\dot{A}^{\prime }\dot{B}%
^{\prime }}.
\end{equation*}%
A corresponding abstract d--spinor calculus for last formula in Theorem \ref%
{thb1} results in proof of

\begin{proposition}
In N--adapted form, $\widehat{\Psi }_{\dot{A}\dot{B}\dot{C}\dot{T}}=\Psi _{%
\dot{A}\dot{B}\dot{C}\dot{T}}.$
\end{proposition}

Applying statements of Theorem \ref{thb2}, we obtain formulas%
\begin{equation*}
\widehat{\mathbf{D}}^{\dot{A}\dot{A}^{\prime }}\Psi _{\dot{A}\dot{B}\dot{C}%
\dot{T}}=\Upsilon ^{\dot{A}\dot{A}^{\prime }}\Psi _{\dot{A}\dot{B}\dot{C}%
\dot{T}}\mbox{ and }\mathbf{D}_{\dot{A}^{\prime }}^{\dot{T}}\Psi _{\dot{A}%
\dot{B}\dot{C}\dot{T}}=\mathbf{D}_{(\dot{B}}^{\dot{B}^{\prime }}\mathbf{P}_{%
\dot{A})\dot{C}\dot{A}^{\prime }\dot{B}^{\prime }},
\end{equation*}%
and $\widehat{\mathbf{D}}_{(\dot{B}}^{\dot{B}^{\prime }}\widehat{\mathbf{%
\Phi }}_{\dot{C}\dot{T})\dot{A}^{\prime }\dot{B}^{\prime }}=\Upsilon _{\dot{A%
}^{\prime }}^{\dot{A}}\Psi _{\dot{A}\dot{B}\dot{C}\dot{T}}$.

\begin{remark}
Vacuum solutions of Einstein equations in general relativity and models of
Finsler--Cartan gravity with $\mathbf{D,}$ or $\mathbf{\tilde{D},}$ are not
conformally invariant. This follows from the fact that even $\mathbf{\Phi }_{%
\dot{C}\dot{T}\dot{A}^{\prime }\dot{B}^{\prime }}=0$ the above formula with $%
\widehat{\mathbf{D}}_{(\dot{B}}^{\dot{B}^{\prime }}\widehat{\mathbf{\Phi }}_{%
\dot{C}\dot{T})\dot{A}^{\prime }\dot{B}^{\prime }}$ does not result in zero $%
\widehat{\mathbf{\Phi }}_{\dot{C}\dot{T}\dot{A}^{\prime }\dot{B}^{\prime }}.$
Nevertheless, it should be emphasized here that such results are for a fixed
N--connection structure $\mathbf{N,}$ or $\mathbf{\tilde{N}.}$ We can
transform some data $(\mathbf{g,N)}$ with (non) zero $\mathbf{\Phi }_{\dot{C}%
\dot{T}\dot{A}^{\prime }\dot{B}^{\prime }}$ into certain $(^{\eta }\mathbf{%
g,^{\eta }N)}$ with, for instance, $\mathbf{^{\eta }\Phi }_{\dot{C}\dot{T}%
\dot{A}^{\prime }\dot{B}^{\prime }}=0,$ and/or $\mathbf{^{\eta }}\widehat{%
\mathbf{D}}_{(\dot{B}}^{\dot{B}^{\prime }}\ \mathbf{^{\eta }}\widehat{%
\mathbf{\Phi }}_{\dot{C}\dot{T})\dot{A}^{\prime }\dot{B}^{\prime }}=\mathbf{%
^{\eta }}\Upsilon _{\dot{A}^{\prime }}^{\dot{A}}\mathbf{^{\eta }}\Psi _{\dot{%
A}\dot{B}\dot{C}\dot{T}}=0,$ for instance, following the anholonomic
deformation method \cite{veyms,vsgg,vrflg}.
\end{remark}

Via N--adapted d--spinor calculus, we can prove

\begin{theorem}
Under N--adapted conformal transforms, the values determining the vacuum
Einstein equations for $\mathbf{D}$ transform following rules%
\begin{eqnarray*}
\widehat{\mathbf{\Phi }}_{\dot{C}\dot{T}\dot{A}^{\prime }\dot{B}^{\prime }}
&=&\mathbf{\Phi }_{\dot{C}\dot{T}\dot{A}^{\prime }\dot{B}^{\prime }}=-%
\mathbf{D}_{\dot{C}(\dot{B}^{\prime }}\Upsilon _{\dot{A}^{\prime })\dot{T}%
}+\Upsilon _{\dot{C}(\dot{B}^{\prime }}\Upsilon _{\dot{A}^{\prime })\dot{T}}
\\
&=&\varpi \mathbf{D}_{\dot{C}(\dot{B}^{\prime }}\mathbf{D}_{\dot{A}^{\prime
})\dot{T}}\varpi ^{-1}=-\varpi ^{-1}\widehat{\mathbf{D}}_{\dot{C}(\dot{B}%
^{\prime }}\widehat{\mathbf{D}}_{\dot{A}^{\prime })\dot{T}}\varpi , \\
4\varpi ^{2}\widehat{\Lambda } &=&4\Lambda +(\mathbf{D}^{\alpha }\Upsilon
_{\alpha }+\Upsilon ^{\alpha }\Upsilon _{\alpha })=4\Lambda +\varpi
^{-1}\square \varpi ,
\end{eqnarray*}%
for $\square :=\mathbf{D}^{\alpha }\mathbf{D}_{\alpha }.$
\end{theorem}

Using the operators $\mathbf{D}_{\dot{C}\dot{B}^{\prime }}$ and/or $\mathbf{%
\tilde{D}}_{\dot{C}\dot{B}^{\prime }},$ we can construct other conformally
N--adapted invariant values, for instance, a Finsler like Bach d--tensor (it
can be constructed similarly to formulas (6.8.42)-(6.8.45) in \cite{penrr}
but for d--connections).

Finally, we note that we can formulate spinor N--adapted differential
geometries and derive following geometric/N--adapted variational principles certain  gravitational and matter field equations for $\mathbf{V},\dim
\mathbf{V}=n+m;n,m\geq 2$, generalizing for higher dimensions the tensor and spinor abstract index formalism.

\section{Nonholonomic (Finsler) twistors}

\label{snft}

 The twistor theory was elaborated with a
very general goal to translate the standard  physics in the language of complex
manifolds mathematics when spacetime points and fundamental properties and
field interactions are derived from certain fundamental principles being
generalized former constructions for spinor algebra and geometry. The
approach is characterized by certain important results in generating exact
solutions of fundamental matter field equations, twistor methods of
quantization, formulating conservation laws in gravity and encoding, for
instance, of (anti) self--dual Yang--Mills and gravitational interactions,
see details in Refs. \cite{penrr,manin,ward,huggett}.

In some sense, our research interests are twofold: 1) The first aim  to understand to what extend the
Twistor Program can be generalized for modified Finsler type spacetime
geometries.  2) Nonholonomic (Finsler) methods happen to be very effective in
elaborating new geometric methods of constructing exact solutions and
quantization of gravity theories. The second aim is to clarify  how such approaches can be related to
spinor and twistor geometry?  Our constructions  should provide not only "pure"
academic generalizations of twistor geometry for "more sophisticate"
spacetime models with local anisotropies. Finsler like variables can be
introduced even in general relativity (similarly to various former tetradic,
spinor etc approaches) which give us new possibilities for developing the
twistor theory for curved spaces and generic off--diagonal gravitational and
matter field interactions.

The aim of this section is to define twistors for nonholonomic (Finsler)
spaces and show how such constructions can be globalized on curved spaces
via nonholonomic deformations of fundamental geometric structures.

\subsection{Twistor equations for nonholonomic 2+2 splitting}

Originally, twistors were introduced for complexified projective models of
flat Minkowski spacetimes using the two--spinor formalism. For nonholonomic
manifolds enabled with N--connection structure, we can consider analogs of
flat spaces determined by data $(\mathbf{g,N,D})$ for which the N--adapted
Riemannian curvature and the conformal Weyl d--tensors are zero (see
Theorem \ref{thb1}). In general, such geometries are curved ones because the
curvature of $\nabla $ is not zero. The spinor constructions are similar to
those for (pseudo) Euclidean spaces if there are used N--adapted frames (\ref%
{dder}) and (\ref{ddif}).

\subsubsection{Definition of nonholonomic  twistors}

Let us consider analogs of flat twistors on spaces enabled with
N--connection structure.

\begin{definition}
The nonholonomic twistor equations are
\begin{equation}
\mathbf{D}_{\dot{A}^{\prime }}^{(\dot{A}}\mathbf{\omega }^{\dot{B})}=0.
\label{nhtweq}
\end{equation}
\end{definition}

Let us formulate the conditions when such equations are conformally
invariant in N--adapted form. Choosing $\widehat{\mathbf{\omega }}^{\dot{B}}=%
\mathbf{\omega }^{\dot{B}},$ we can compute
\begin{equation}
\widehat{\mathbf{D}}_{\dot{A}\dot{A}^{\prime }}\widehat{\mathbf{\omega }}^{%
\dot{B}}=\widehat{\mathbf{D}}_{\dot{A}\dot{A}^{\prime }}\widehat{\mathbf{%
\omega }}^{\dot{B}}+\delta _{\dot{A}}^{\dot{B}}\Upsilon _{\dot{C}\dot{A}%
^{\prime }}\mbox{ and }\widehat{\mathbf{D}}_{\dot{A}^{\prime }}^{(\dot{A}}%
\widehat{\mathbf{\omega }}^{\dot{B})}=\varpi ^{-1}\mathbf{D}_{\dot{A}%
^{\prime }}^{(\dot{A}}\mathbf{\omega }^{\dot{B})}  \label{confntr}
\end{equation}%
where $\Upsilon _{\dot{C}\dot{A}^{\prime }}$ is given by the conformal
d--vector (\ref{confdvect}).\footnote{%
For simplicity, we shall consider that the spinor $\mathbf{\omega }^{\dot{B}%
} $ does not posses an electromagnetic charge.}

\begin{lemma}
\ The nonholonomic twistor equations (\ref{nhtweq}) are compatible if and
only if
\begin{equation}
\Psi _{\dot{T}\dot{A}\dot{B}\dot{C}}\mathbf{\omega }^{\dot{T}}=0.
\label{nhtwcomp}
\end{equation}
\end{lemma}

\begin{proof}
It follows from $\mathbf{D}^{\dot{A}^{\prime }(\dot{C}}\mathbf{D}_{\dot{A}%
^{\prime }}^{\dot{A}}\mathbf{\omega }^{\dot{B})}=-\square ^{(\dot{C}\dot{A}}%
\mathbf{\omega }^{\dot{B})}=-\Psi _{\quad \dot{T}}^{\dot{A}\dot{B}\ \dot{C}}%
\mathbf{\omega }^{\dot{T}}.$

$\square $
\end{proof}

Using this lemma and via straightforward verifications in N--adapted frames,
we can prove

\begin{theorem}
If the compatibility conditions (\ref{nhtwcomp}) are satisfied, we can solve
the nonholonomic twistor equations (\ref{nhtweq}) in general form,%
\begin{eqnarray}
\mathbf{\omega }^{\dot{B}} &=&\mathbf{\mathring{\omega}}^{\dot{B}}-iu^{^{%
\dot{B}\dot{B}\prime }}\mathring{\pi}_{\dot{B}\prime }\mbox{ and }\pi _{\dot{%
B}\prime }=\mathring{\pi}_{\dot{B}\prime },  \label{nhtsol} \\
\mathbf{D}_{\dot{B}\dot{B}\prime }\mathbf{\omega }^{\dot{C}} &=&-i\delta _{%
\dot{B}}^{\dot{C}}\pi _{\dot{B}\prime },  \notag
\end{eqnarray}%
where the point $u^{^{\dot{B}\dot{B}\prime }}\in \mathbf{V,}i^{2}=-1,$ and $%
\mathbf{\mathring{\omega}}^{\dot{B}}$ and $\mathring{\pi}_{\dot{B}\prime }$
are constant values with respect to N--adapted frames (\ref{dder}) and (\ref%
{ddif}) for which $\Psi _{\dot{A}\dot{B}\dot{C}\dot{T}}=0.$
\end{theorem}

We can generalize the concept of twistors for flat spaces to nonholonomic
manifolds which are conformally flat in N--adapted form:

\begin{definition}
An nonholonomic (equivalently, anholonomic) twisor space $\mathbb{T}^{\dot{%
\alpha}}$ is a four dimensional complex vector space (with real eight
dimensions) determined by elements of type $\boldsymbol{Z}^{\dot{\alpha}}=(%
\mathbf{\omega }^{\dot{A}},\pi _{\dot{A}^{\prime }})$ with the two spinor
components $\boldsymbol{Z}^{\dot{A}}=\mathbf{\omega }^{\dot{A}}$ and $%
\boldsymbol{Z}_{\dot{A}^{\prime }}=\pi _{\dot{A}^{\prime }}$ defined by
solutions of type (\ref{nhtsol}).
\end{definition}

Doted indices are used in order to emphasize that we work in N--adapted
form. This allows us to preserve with respect to N--adapted frames certain
similarity to formulas from \cite{penrr}. If $\mathbf{D\rightarrow \nabla ,}$
for conformally flat spaces, we obtain $\boldsymbol{Z}^{\dot{\alpha}%
}\rightarrow Z^{\alpha }\in \mathbb{T}^{\alpha },$ i.e. standard Penrose's
twistors. Such constructions and relations depends on a point $\ ^{0}u$
fixing a coordinate system.

\begin{remark}
In a similar form, we can introduce the space of nonholonomic dual twistors $%
\mathbb{T}_{\dot{\alpha}}$ $\ $with elements $\boldsymbol{W}_{\dot{\alpha}}=(%
\boldsymbol{W}_{\dot{A}}=\mathbf{\lambda }_{\dot{A}},\boldsymbol{W}^{\dot{A}%
^{\prime }}=\mu ^{\dot{A}^{\prime }}),$ where
\begin{eqnarray}
\mathbf{\lambda }_{\dot{A}} &=&\mathbf{\mathring{\lambda}}_{\dot{A}}%
\mbox{
and }\mu ^{\dot{A}^{\prime }}=\mathring{\mu}^{\dot{A}^{\prime }}+iu^{^{\dot{A%
}\dot{A}^{\prime }}}\mathbf{\mathring{\lambda}}_{\dot{A}},  \label{nhtsold}
\\
\mathbf{D}_{\dot{A}\dot{A}^{\prime }}\mu ^{\dot{B}\prime } &=&i\delta _{\dot{%
A}^{\prime }}^{\dot{B}\prime }\mathbf{\lambda }_{\dot{A}},  \notag
\end{eqnarray}%
are  solutions of the dual nonholonomic twistor equations $\mathbf{%
D}_{\dot{A}}^{(\dot{A}^{\prime }}\mu ^{\dot{B}\prime )}=0$.
\end{remark}

The complex conjugation of nonholonomic (dual) twistors follows the rules%
\begin{equation*}
\overline{\boldsymbol{Z}^{\dot{\alpha}}}=\overline{\boldsymbol{Z}}_{\dot{%
\alpha}}:=(\overline{\pi }_{\dot{A}},\overline{\mathbf{\omega }}^{\dot{A}%
^{\prime }})\mbox{ and }\overline{\boldsymbol{W}_{\dot{\alpha}}}=\overline{%
\boldsymbol{W}}^{\dot{\alpha}}:=(\overline{\mu }^{\dot{A}},\overline{\mathbf{%
\lambda }}_{\dot{A}^{\prime }}).
\end{equation*}%
We can consider higher valence twistors, for instance, $\mathbf{X}_{\ \dot{%
\beta}}^{\dot{\alpha}}$ where N--adapted twistor indices transform
respectively following rules (\ref{nhtsol}) and (\ref{nhtsold}) taken "-" or
"+" before complex unity $i.$

\subsubsection{Geometric/physical meaning of anholonomic twistors}

A nonholonomic frame structure prescribes a corresponding spiral
configuration for twistors and their conformal transforms.

\begin{definition}
\textbf{-Corollary:} The class of curved spaces generated by anholonomy
relations (\ref{anhrel}) subjected to the compatibility conditions (\ref%
{nhtwcomp}) is characterized by anholonomic spirality%
\begin{equation}
\mathbf{\dot{s}}:=\frac{1}{2}\boldsymbol{Z}^{\dot{\alpha}}\overline{%
\boldsymbol{Z}}_{\dot{\alpha}},  \label{spiralityd}
\end{equation}
which is invariant under N--adapted conformal transforms.
\end{definition}

\begin{proof}
It follows from verification that $\boldsymbol{Z}^{\dot{\alpha}}%
\overline{\boldsymbol{Z}}_{\dot{\alpha}}=\widehat{\boldsymbol{Z}}^{\dot{%
\alpha}}\widehat{\overline{\boldsymbol{Z}}}_{\dot{\alpha}}$ (using formulas (%
\ref{confntr}) and (\ref{nhtsol})).

$\square $
\end{proof}

In both holonomic and nonholonomic cases, the simplest geometric
interpretation is possible for the so--called isotropic twistors when $%
\boldsymbol{Z}^{\dot{\alpha}}\overline{\boldsymbol{Z}}_{\dot{\alpha}}=0.$
Fixing a value $\mathring{\pi}_{\dot{B}\prime }\neq 0,$ we get from (\ref%
{nhtsol}) that%
\begin{equation*}
u^{\dot{B}\dot{B}\prime }=\mathbf{\mathring{\omega}}^{\dot{B}}\overline{%
\mathbf{\mathring{\omega}}}^{\dot{B}^{\prime }}/i\overline{\mathbf{\mathring{%
\omega}}}^{\dot{A}^{\prime }}\mathring{\pi}_{\dot{A}^{\prime }}+\tau
\overline{\pi }^{\dot{B}}\pi ^{\dot{B}\prime },\tau \in \mathbb{R},
\end{equation*}%
describes a light ray propagating in N-adapted form in an effective locally
anisotropic media and/or a curved spacetime with geometric objects induced
by nontrivial anholonomy coefficients . If $\mathring{\pi}_{\dot{B}\prime
}=0,$ such a light ray is moved to infinity.

We can also characterize massless particles with momentum, rotation and
spirality propagating in effective curved spaces derived for certain
anholonomy relations of moving frames. Taking $\boldsymbol{Z}^{\dot{\alpha}%
}=(\mathbf{\omega }^{\dot{A}},\pi _{\dot{A}^{\prime }})$ with $\pi _{\dot{A}%
^{\prime }}\neq 0,$ we construct
\begin{equation*}
p_{\dot{A}^{\prime }\dot{A}^{\prime }}:=\overline{\pi }_{\dot{A}}\pi _{\dot{A%
}^{\prime }},M^{\dot{A}\dot{A}^{\prime }\dot{B}\dot{B}^{\prime }}:=i\mathbf{%
\omega }^{(\dot{A}}\overline{\pi }^{\dot{B})}\varepsilon ^{\dot{A}^{\prime }%
\dot{B}\prime }-i\overline{\mathbf{\omega }}^{(\dot{A}^{\prime }}\pi ^{\dot{B%
}^{\prime })}\varepsilon ^{\dot{A}\dot{B}}
\end{equation*}%
and spin d--vector $S_{\alpha }=\frac{1}{2}e_{\alpha \beta \gamma \tau
}p^{\beta }M^{\gamma \tau }=\mathbf{\dot{s}}p_{\alpha }$, where $e_{\alpha
\beta \gamma \tau }$ is the absolute antisymmetric d--tensor and $\mathbf{%
\dot{s}}$ is computed as in (\ref{spiralityd}). In local N--adapted form,
such a physical interpretation of nonholonomic twistors is possible with
respect to bases of type (\ref{dder}) and (\ref{ddif}) for which $\Psi _{%
\dot{A}\dot{B}\dot{C}\dot{T}}=0.$ This describes a massless particle moving in
a subclass of curved spaces with nontrivial curvature for $\nabla $ when
certain anholonomic constraints are imposed.

\subsection{Finsler twistors on tangent bundles}

Originally, the Finsler--Cartan geometry was constructed on tangent bundles
with $\boldsymbol{D=\tilde{D}}$ and $4+4$ splitting as we explained in
section \ref{ssmfcs}. The corresponding Weyl d--tensor $\mathbf{\tilde{C}}%
_{\tau \alpha \beta \gamma }$ is computed using formulas (\ref{wtcd}) but
for curvature coefficients (\ref{curvcart}) and curvature spinors (\ref%
{curvfcsp}).

\begin{definition}
The twistor equations for Finsler--Cartan geometries are
\begin{equation}
\mathbf{D}_{I^{\prime }}^{(I}\mathbf{\omega }^{J)}=0,\mathbf{D}_{A^{\prime
}}^{(A}\mathbf{\omega }^{B)}=0.  \label{fctweq}
\end{equation}
\end{definition}

Choosing $\widehat{\mathbf{\omega }}^{J}=\mathbf{\omega }^{J},\widehat{%
\mathbf{\omega }}^{B}=\mathbf{\omega }^{B},$ we can compute
\begin{eqnarray*}
\widehat{\mathbf{D}}_{II^{\prime }}\widehat{\mathbf{\omega }}^{J} &=&%
\widehat{\mathbf{D}}_{II^{\prime }}\widehat{\mathbf{\omega }}^{J}+\delta
_{I}^{J}\tilde{\Upsilon}_{KJ^{\prime }}\mbox{ and }\widehat{\mathbf{D}}%
_{I^{\prime }}^{(I}\widehat{\mathbf{\omega }}^{J)}=\varpi ^{-1}\mathbf{D}%
_{I^{\prime }}^{(I}\mathbf{\omega }^{J)}, \\
\widehat{\mathbf{D}}_{AA^{\prime }}\widehat{\mathbf{\omega }}^{B} &=&%
\widehat{\mathbf{D}}_{AA^{\prime }}\widehat{\mathbf{\omega }}^{B}+\delta
_{A}^{B}\tilde{\Upsilon}_{CA^{\prime }}\mbox{ and }\widehat{\mathbf{D}}%
_{A^{\prime }}^{(A}\widehat{\mathbf{\omega }}^{B)}=\varpi ^{-1}\mathbf{D}%
_{A^{\prime }}^{(A}\mathbf{\omega }^{B)\prime },
\end{eqnarray*}%
where $\tilde{\Upsilon}_{II^{\prime }}:=\mathbf{\tilde{D}}_{II^{\prime }}\ln
\varpi $ $\Upsilon _{\dot{C}\dot{A}^{\prime }},\tilde{\Upsilon}_{AA^{\prime
}}:=\mathbf{\tilde{D}}_{AA^{\prime }}\ln \varpi $ are constructed similarly
to the conformal d--vector (\ref{confdvect}).

If $\boldsymbol{D=\tilde{D},}$ we can obtain from the Theorem \ref{thb1} the

\begin{corollary}
The anti--sefldual Weyl d--spinors corresponding to the Cartan d--curvature \newline $%
\mathbf{\tilde{R}}_{\ \beta \gamma \tau }^{\alpha }=\{\tilde{R}_{\ hjk}^{i},%
\tilde{P}_{\ jka}^{i},\tilde{S}_{\ bcd}^{a}\}$ (\ref{curvcart}) are
characterized by $h$-- and $v$--components $\{\tilde{\Psi}_{LIJK},\tilde{\Psi%
}_{DIJK},\tilde{\Psi}_{DABC}\}.$
\end{corollary}

This results in a set of three conditions of compatibility:

\begin{lemma}
 The Finsler--Cartan twistor equations (\ref{nhtweq}) are compatible if and
only if
\begin{equation}
\tilde{\Psi}_{LIJK}\mathbf{\omega }^{L}=0,\tilde{\Psi}_{DIJK}\mathbf{\omega }%
^{D}=0,\tilde{\Psi}_{DABC}\mathbf{\omega }^{D}=0.  \label{fctwcomp}
\end{equation}
\end{lemma}

All results on Finsler--Cartan twistors can be proved using formal Sasaki
lifts $\mathbf{g\rightarrow \tilde{g}}$ (\ref{slm}) and $\mathbf{%
N\rightarrow \tilde{N}}$ (\ref{cncl}) with spinor coefficients (\ref%
{dmetrspindec}) for the constructions with nonholonomic twistors and
canonical d--connections.

\begin{theorem}
If the compatibility conditions (\ref{fctwcomp}) are satisfied, we can solve
the nonholonomic twistor equations (\ref{fctweq}) in general form for $h$%
--components,%
\begin{eqnarray}
\mathbf{\omega }^{J} &=&\mathbf{\mathring{\omega}}^{J}-iu^{^{JJ\prime }}%
\mathring{\pi}_{J\prime }\mbox{ and }\pi _{J\prime }=\mathring{\pi}_{J\prime
},  \label{htfcso} \\
\mathbf{D}_{JJ\prime }\mathbf{\omega }^{K} &=&-i\delta _{J}^{K}\pi _{J\prime
},  \notag
\end{eqnarray}%
and for $v$--components%
\begin{eqnarray}
\mathbf{\omega }^{B} &=&\mathbf{\mathring{\omega}}^{B}-iu^{^{BB\prime }}%
\mathring{\pi}_{B\prime }\mbox{ and }\pi _{B\prime }=\mathring{\pi}_{B\prime
},  \label{vtfcso} \\
\mathbf{D}_{BB\prime }\mathbf{\omega }^{C} &=&-i\delta _{B}^{C}\pi _{B\prime
},  \notag
\end{eqnarray}%
where the point $u^{\alpha }=(u^{^{II\prime }},u^{^{AA\prime }})\in \mathbf{%
TM}$ and constant values are considered with respect to N--adapted frames (%
\ref{dder}) and (\ref{ddif}) when the conditions (\ref{fctwcomp}) are
satisfied.
\end{theorem}

We can generalize the concept of twistors for flat spaces to nonholonomic
manifolds which are conformally flat in N--adapted form:

\begin{definition}
\begin{itemize}
\item A horizontal twisor space $h\mathbb{T}^{\dot{\alpha}}$ is a four
dimensional complex vector space (with real eight dimensions) determined by
elements of type $h\boldsymbol{Z}^{\dot{i}}=(\mathbf{\omega }^{I},\pi
_{I^{\prime }}),$ $\dot{i}=1,2,3,4$ with the two spinor components $h%
\boldsymbol{Z}^{I}=\mathbf{\omega }^{I}$ and $\boldsymbol{Z}_{I^{\prime
}}=\pi _{I^{\prime }}$ defined by solutions of type (\ref{htfcso}).

\item A vertical twisor space $v\mathbb{T}^{\dot{\alpha}}$ is a four
dimensional complex vector space (with real eight dimensions) determined by
elements of type $h\boldsymbol{Z}^{\dot{a}}=(\mathbf{\omega }^{A},\pi
_{A^{\prime }}),\dot{a}=5,6,7,8$ with the two spinor components $v%
\boldsymbol{Z}^{A}=\mathbf{\omega }^{A}$ and $\boldsymbol{Z}_{A^{\prime
}}=\pi _{A^{\prime }}$ defined by solutions of type (\ref{vtfcso}).
\end{itemize}
\end{definition}

We conclude that nonholonomic twistor constructions for the Finsler--Cartan
spaces dub as $h$- and $v$--components the values introduced via canonical
d--connections on $\mathbf{V}$. Re--defining the abstract index formalism
for d--tensors and d--spinors, all formulas can be proved by similarity in
N--adapted frames.

\subsection{Nonholonomic local and global twistors}

On N--adapted conformally flat nonholonomic manifolds, the solutions (\ref%
{nhtsol}) of generalized twistor equations (\ref{nhtweq}) define certain
global anholonomic twistor structures. If the conditions (\ref{fctwcomp}) are
not satisfied, we can only define a nonholonomic twistor bundle on a $%
\boldsymbol{V}$ when the geometric object depend on base manifold points. This
does not define an alternative description of nonholonomic manifolds (and
Finsler--Cartan geometries) \ in terms of certain generalized nonholonomic
twistor spaces. For a prescribed N--connection structure $N,$ we can
construct N--adapted local twistors with properties similar to those of
holonomic twistors considered in Chapter 6, paragraph 9, in \cite{penrr}.

Nevertheless, nonholonomic/ Finsler spaces are characterized by more rich
geometric structures which provide us new possibilities and methods for
constructing new classes of generalized twistor -- Finsler spaces and
applications in general relativity and modifications. We study two models of
N--adapted twistor spaces in local and global forms.

\subsubsection{N--adapted local twistors and torsionless conditions}

Let us consider a point $u\in \mathbf{V}$ for nonholonomic data $(\mathbf{%
g,N,D}).$

\begin{definition}
A local N--adapted twistor $\ _{u}\mathbf{Z}^{\dot{\alpha}}$ (in brief, local
d--twistor) in a point $u$ is given by a couple of N--adapted two--spinors $%
(\ _{u}\mathbf{\omega }^{\dot{A}},\ _{u}\pi _{\dot{A}^{\prime }})$ in this
point, which in a chosen anholonomic frame (\ref{dder}) and (\ref{ddif}) satisfied the rules: if $\mathbf{g}_{\alpha \beta }\rightarrow \widehat{%
\mathbf{g}}_{\alpha \beta }:=\varpi ^{2}\mathbf{g}_{\alpha \beta }$ then
\begin{equation*}
\ _{u}\mathbf{Z}^{\dot{\alpha}}=(\ _{u}\mathbf{\omega }^{\dot{A}},\pi _{\dot{%
A}^{\prime }})\rightarrow \ _{u}\widehat{\mathbf{Z}}^{\dot{\alpha}}=(\ _{u}%
\widehat{\mathbf{\omega }}^{\dot{A}}=\ _{u}\mathbf{\omega }^{\dot{A}},\ _{u}%
\widehat{\pi }_{\dot{A}^{\prime }}=\ _{u}\pi _{\dot{A}^{\prime }}+i\Upsilon
_{\dot{A}\dot{A}^{\prime }}\ _{u}\mathbf{\omega }^{\dot{A}}).
\end{equation*}
\end{definition}

The local d--vectors $\ _{u}\mathbf{Z}^{\dot{\alpha}}$ and $\ _{u}\widehat{%
\mathbf{Z}}^{\dot{\alpha}}$ depend functionally, respectively, on $(\mathbf{%
g,N,D,}\varpi ).$ The set $\mathbf{Z}^{\dot{\alpha}}=$ $\cup _{u\in \mathbf{V%
}}(\ _{u}\mathbf{Z}^{\dot{\alpha}})$ of all local twistors $\ _{u}\mathbf{Z}%
^{\dot{\alpha}}$ taken in all points $u$ of $\mathbf{V}$ defines a vector
bundle, when the fiber in $u$ is a complex four dimensional vector space (i.
e. the spaces of local N--adapted twistors in $u$). Such a vector bundle is
nonholonomic being endowed with N--connection structure. For simplicity, we
shall omit the left low label "u" and write a local twistor as $\mathbf{Z}^{%
\dot{\alpha}}$ if that will not result in ambiguities.

In N--adapted (and/or general local) form, the connection $\nabla $ can be constructed to possess
zero coefficients in a point and/or along a curve though such a point (the
so--called normal coordinates). This allows us to define transports of usual
local twistors along curves with tangent vector fields $t^{\dot{A}\dot{A}%
^{\prime }}\subset T\mathbf{V}.$ We can generalize such formulas for
d--connections $\mathbf{D}$ and d--vectors $\mathbf{t}^{\dot{A}\dot{A}%
^{\prime }}$ and consider a local d--twistor $\mathbf{Z}^{\dot{\alpha}}$
which is constant in N--adapted direction $\mathbf{t}^{\alpha },$ when%
\begin{eqnarray}
t^{\dot{A}\dot{A}^{\prime }}\mathbf{D}_{\dot{A}\dot{A}^{\prime }}\mathbf{%
\omega }^{\dot{B}}+it^{\dot{B}\dot{A}^{\prime }}\pi _{\dot{A}^{\prime }}
&=&0,  \label{transpdtw} \\
t^{\dot{A}\dot{A}^{\prime }}\mathbf{D}_{\dot{A}\dot{A}^{\prime }}\pi _{\dot{B%
}^{\prime }}+it^{\dot{A}\dot{A}^{\prime }}\mathbf{P}_{\dot{A}\dot{A}^{\prime
}\dot{B}\dot{B}^{\prime }}\mathbf{\omega }^{\dot{B}} &=&0,  \notag
\end{eqnarray}%
where $\mathbf{P}_{\dot{A}\dot{A}^{\prime }\dot{B}\dot{B}^{\prime }}$ is
related to the Ricci d--tensor and scalar curvature of $\mathbf{D}$ as in
formula (\ref{pdt}) and a curve is defined in vicinity of a point $\ ^{0}u$
in the form $u^{\dot{A}\dot{A}^{\prime }}(\tau )=\ ^{0}u^{\dot{A}\dot{A}%
^{\prime }}+t^{\dot{A}\dot{A}^{\prime }}\tau ,$ for a real parameter $\tau .$
The d--twistor transport equations (\ref{transpdtw}) have constant twistor
solutions in $\ ^{0}u$ and along $u(\tau )$ which in any point satisfy the
conditions
\begin{equation}
\mathbf{D}_{\dot{A}^{\prime }}^{(\dot{A}}\mathbf{\omega }^{\dot{B})}=0%
\mbox{
and }\pi _{\dot{A}^{\prime }}=\frac{1}{2}i\mathbf{D}_{\dot{A}\dot{A}^{\prime
}}\mathbf{\omega }^{\dot{A}}.  \label{solparaldtr}
\end{equation}

\begin{corollary}
The solutions for local d--twistors (\ref{solparaldtr}) can be globalized on
$\mathbf{V}$ if and only if the conditions (\ref{nhtwcomp}) are satisfied,
for instance, if the Weyl d--spinor vanishes.
\end{corollary}

\begin{definition}
The N--adapted covariant derivative operator (d--connection) along $\mathbf{t%
}^{\alpha }$ in the space of local d--twistors is by definition
\begin{equation}
\ _{\mathbf{t}}\mathbf{D:=}t^{\dot{A}\dot{A}^{\prime }}\mathbf{D}_{\dot{A}%
\dot{A}^{\prime }}  \label{ltdcon}
\end{equation}
\end{definition}

The local d--twistor d--connection (\ref{ltdcon}) allows us to compute the
variation of $\mathbf{Z}^{\dot{\alpha}}$ along $u(\tau ),$ following
formulas (\ref{transpdtw}) with nonzero right sides,
\begin{equation*}
\ _{\mathbf{t}}\mathbf{DZ}^{\dot{\alpha}}=\left( t^{\dot{A}\dot{A}^{\prime }}%
\mathbf{D}_{\dot{A}\dot{A}^{\prime }}\mathbf{\omega }^{\dot{B}}+it^{\dot{B}%
\dot{A}^{\prime }}\pi _{\dot{A}^{\prime }},t^{\dot{A}\dot{A}^{\prime }}%
\mathbf{D}_{\dot{A}\dot{A}^{\prime }}\pi _{\dot{B}^{\prime }}+it^{\dot{A}%
\dot{A}^{\prime }}\mathbf{P}_{\dot{A}\dot{A}^{\prime }\dot{B}\dot{B}^{\prime
}}\mathbf{\omega }^{\dot{B}}\right) .
\end{equation*}%
Here we note that in similar form we can define dual local d--twistors of
type $\boldsymbol{W}_{\dot{\alpha}}=(\mathbf{\lambda }_{\dot{A}},\mu ^{\dot{A%
}^{\prime }})$ with N--adapted conformally invariant scalar product
\begin{equation*}
\boldsymbol{W}_{\dot{\alpha}}\mathbf{Z}^{\dot{\alpha}}:=\mathbf{\lambda }_{%
\dot{A}}\mathbf{\omega }^{\dot{A}}+\mu ^{\dot{A}^{\prime }}\pi _{\dot{A}%
^{\prime }}
\end{equation*}%
and property that
\begin{equation*}
\ _{\mathbf{t}}\mathbf{D(\boldsymbol{W}_{\dot{\alpha}}\mathbf{Z}^{\dot{\alpha%
}})=(\ _{\mathbf{t}}\mathbf{D}\boldsymbol{W}_{\dot{\alpha}})\mathbf{Z}^{\dot{%
\alpha}}+\boldsymbol{W}_{\dot{\alpha}}(\ _{\mathbf{t}}\mathbf{DZ}^{\dot{%
\alpha}}).}
\end{equation*}

\begin{definition}
\textbf{-Lemma}:\ The curvature d--tensor of local d--twistor d--con\-nection is
 $$i\left(\ _{\mathbf{t}}\mathbf{D}\ _{\mathbf{v}}\mathbf{D-}\ _{\mathbf{v}}%
\mathbf{D}\ _{\mathbf{t}}\mathbf{D-}\ _{[\mathbf{t,v]}}\mathbf{D}\right)
\mathbf{\mathbf{Z}^{\dot{\beta}}} =\mathbf{Z}^{\dot{\alpha}}\mathbf{K}_{\
\dot{\alpha}}^{\dot{\beta}\ }(\mathbf{t},\mathbf{v}) =\mathbf{t}^{\mu }\mathbf{v}^{\nu }\mathbf{K}_{\ \dot{\alpha}\mu \nu }^{%
\dot{\beta}},$$
 for two d--vectors $\mathbf{t}^{\mu }=\mathbf{t}^{\dot{M}\dot{M}^{\prime }},%
\mathbf{v}^{\nu }=\mathbf{v}^{\dot{N}\dot{N}^{\prime }}\in T\mathbf{V}$ and
computed with N--adapted coefficients,%
\begin{equation*}
\mathbf{K}_{\ \dot{\alpha}\mu \nu }^{\dot{\beta}}=\left(
\begin{array}{cc}
i\varepsilon _{\dot{M}^{\prime }\dot{N}^{\prime }}\mathbf{\Psi }_{\ \dot{A}%
\dot{M}\dot{N}}^{\dot{B}} & \varepsilon _{\dot{M}\dot{N}}\mathbf{D}_{\dot{A}%
}^{\dot{A}^{\prime }}\overline{\mathbf{\Psi }}_{\ \dot{A}^{\prime }\dot{M}%
^{\prime }\dot{N}^{\prime }}^{\dot{B}^{\prime }}+\varepsilon _{\dot{M}%
^{\prime }\dot{N}^{\prime }}\mathbf{D}_{\dot{B}^{\prime }}^{\dot{B}}\mathbf{%
\Psi }_{\ \dot{A}\dot{M}\dot{N}}^{\dot{B}} \\
0 & -i\varepsilon _{\dot{M}\dot{N}}\overline{\mathbf{\Psi }}_{\ \dot{A}%
^{\prime }\dot{M}^{\prime }\dot{N}^{\prime }}^{\dot{B}^{\prime }}%
\end{array}%
\right) .
\end{equation*}
\end{definition}

Above constructions are determined by data $(\mathbf{g,N,D=}\nabla +\mathbf{Q%
}),$ see (\ref{distrel}). In general, they can be redefined for data $(%
\mathbf{g,}\nabla )$ using nonholonomic deformations.

\begin{theorem}
We can globalize in nonholonomic form the local twistor constructions for $\nabla $ if there is a
N--connection structure $\mathbf{N}$ and associated $\mathbf{D}$ for which $%
\mathbf{K}_{\ \dot{\alpha}\mu \nu }^{\dot{\beta}}=0$ and $\mathbf{\Delta }%
_{\alpha \beta }=0.$
\end{theorem}

\begin{proof}
It is a consequence of conditions of Theorem \ref{thb1} when $\mathbf{\Delta
}_{\alpha \beta }=0$ are equivalent to (\ref{lccond}), i.e. the
nonholonomically induced torsion (by ($\mathbf{g,N))}$ became zero. This is
compatible with Conclusion \ref{concla} when the conformal Weyl
d--tensor/d--spinor for $\mathbf{D}$ can be zero but similar values for $%
\nabla $ are not trivial. Such linear connections are different even in some
N--adapted frames they can be characterized by the same set of coefficients
(transformation laws under frame/coordinate changing are different).

$\square $
\end{proof}

We conclude that via nonholonomic transforms we can generate some compatible
global nonholonomic twistor equations even the standard twistor equations
are not compatible for general curved spacetimes.

\subsubsection{Global extensions of N--adapted twistor structures}

Our idea is to play with such nonholonomic distributions $N$ which allows us
to define spinors and twistors in very general forms.

\begin{claim}
For any (pseudo) Riemannian metric structure $\mathbf{g}$ on a manifold $V$
(or a fundamental Finsler function on tangent bundle $TM),$ we can prescribe a N--connection which allows us
to globalize N--adapted local twistor structures.
\end{claim}

\begin{proof}
Let us fix a d--metric (\ref{m1}) with coefficients $\mathbf{g}_{\alpha
^{\prime }\beta ^{\prime }}:=\varpi ^{2}(u)\mathbf{\eta }_{\alpha \beta }$
with $\mathbf{\eta }_{\alpha \beta }$ being diagonal constants of any
necessary signature $(\pm 1,\pm ,...,\pm ),$ with respect to some $\mathbf{e}%
^{\beta ^{\prime }}=(e^{j^{\prime }}=dx^{j^{\prime }},\mathbf{e}^{b^{\prime
}}=dy^{b^{\prime }}+N_{i^{\prime }}^{b^{\prime }}dx^{i^{\prime }}).$ For
such a d--metric and N--adapted co--bases, we can verify that $\mathbf{C}%
_{\tau ^{\prime }\alpha ^{\prime }\beta ^{\prime }\gamma ^{\prime }}=0$ as
a consequence of Proposition \ref{prop1}. We can redefine data ( (\ref{m1}),
$\mathbf{g}_{\alpha ^{\prime }\beta ^{\prime }})$ in a coordinate form (\ref%
{metr}) \ with coefficients (\ref{ansatz}) (with primed indices, $g_{%
\underline{\alpha }^{\prime }\underline{\beta }^{\prime }}$). Then
considering arbitrary frame transforms $e_{\ \underline{\alpha }}^{%
\underline{\alpha }^{\prime }}$ we can compute $g_{\underline{\alpha }%
\underline{\beta }}=e_{\ \underline{\alpha }}^{\underline{\alpha }^{\prime
}}e_{\ \underline{\beta }}^{\underline{\beta }^{\prime }}g_{\underline{%
\alpha }^{\prime }\underline{\beta }^{\prime }}.$ Finally, we can re--define
for a nonholonomic $2+2,$ or $4+4,$ splitting certain data $\left( \mathbf{g}%
_{\alpha \beta },N_{i}^{a}\right) ,$ for which, in general, $\mathbf{C}%
_{\tau \alpha \beta \gamma }\neq 0,$ and the corresponding to $\nabla
,C_{\tau \alpha \beta \gamma }\neq 0.$ Such construction with nonholonomic
deformations are possible because vierbeins (\ref{vbt}) may depend on some $N
$--coefficients of a generic off--diagonal form of "primary" metric. The
transformation laws of d--objects on nonholonomic manifolds with
N--connection are different from those on usual manifolds without
N--connection splitting (\ref{whitney}).

$\square $
\end{proof}

Nonholonomic twistor spaces can be associated to any metric structure if a necessary
type $h$- $v$--splitting is defined by corresponding N--connections. One
of the important tasks is to formulate such conditions when certain
nonholonomic deformations can be used for encoding exact solutions of
Einstein equations in nonholonomic twistor structures and, inversely, to
formulate nonholonomic twistor transforms generating exact solutions in
general relativity and modifications. We shall provide such constructions in our
further works.


\begin{thebibliography}{99}
\bibitem{penrr} R.\ Penrose and W. Rindler, Spinors and Space--Time, vols. 1
\& 2 (Cambridge University Press, 1984 \& 1986)

\bibitem{manin} Yu. I. Manin, Gauge Field Theory and Complex Geometry
(Springer--Verlag, 1988)

\bibitem{ward} R. S. Ward and R. O. Wells, Jr, Twistor Geometry and Field
Theory (Cambridge University Press, 1990)

\bibitem{huggett} S. A. Huggett and K. P. Tod, An Introduction to Twistor
Theory. London Mathematical Society Student Texts 4, second edition
(Cambridge University Press, 1994)




\bibitem{veyms} S. Vacaru, Decoupling of Field Equations in Einstein and Modified Gravity, J. Phys.: Conf. Ser. \textbf{ 543 } (2013) 012021; arXiv: 1108.2022v3


\bibitem{vbraneq} S. Vacaru, Branes and quantization for an A-model
complexification of Einstein gravity in almost Kaehler variables, Int. J.
Geom. Meth. Mod. Phys. \textbf{\ 6 } (2009) 873-909

\bibitem{vdefq} S. Vacaru, Einstein gravity as a nonholonomic almost K\"{a}%
hler geometry, Lagrange-Finsler variables, and deformation quantization, J.
Geom. Phys. \textbf{\ 60 }\ (2010) 1289-1305

\bibitem{vgauge} S. Vacaru, Two-connection renormalization and nonholonomic
gauge models of Einstein gravity, Int. J. Geom. Meth. Mod. Phys. \textbf{\ 7}%
\ (2010) 713-744

\bibitem{rund} H. Rund, The Differential Geometry of Finsler Spaces
(Springer--Verlag, 1959)

\bibitem{matsumoto} M. Matsumoto, Foundations of Finsler Geometry and
Special Finsler Spaces (Kaisisha: Shigaken, Japan, 1986)


\bibitem{bejancu1} A. Bejancu, Finsler Geometry and Applications (Ellis
Horwood, Chichester, England, 1990)

\bibitem{bejancu2} A. Bejancu and H. R. Farran, Geometry of Pseudo--Finsler
Submanifolds (Kluwer Academic Publishers, 2000)

\bibitem{bcs} D. Bao, S. -S. Chern and Z. Shen, An Introduction to
Riemann--Finsler Geoemtry. Graduate Text in Math., 200 (Springer-Verlag,
2000)

\bibitem{vsgg} Clifford and Riemann- Finsler Structures in Geometric
Mechanics and Gravity, Selected Works, by S. Vacaru, P. Stavrinos, E.
Gaburov and D. Gon\c ta. Differential Geometry - Dynamical Systems,
Monograph 7 (Geometry Balkan Press, 2006),
www.mathem.pub.ro/dgds/mono/va-t.pdf and gr-qc/0508023

\bibitem{stavrinos} P. Stavrinos, Weak gravitational field in
Finsler--Randers space and Raychaudhuri equation, Gen. Rel. Grav.  \textbf{ 44 } (2012) 3029-3045

\bibitem{stavrinosk} A. P. Kouretsis, M. Stathakopoulos and P. C. Stavrinos,
Imperfect fluids, Lorentz violations and Finsler csomology, Phys. Rev. D
\textbf{\ 82 } (2010) 064035

\bibitem{mavromatos} N. E. Mavromatos, V. A. Mitsou, S. Sarkar and A.
Vergou, Implications of a stochastic microscopic Finsler cosmology, Eur.
Phys. J. \textbf{\ C72 } (2012) 1956


\bibitem{visser} J. Skakala and M. Visser, Bi--metric pseudo--Finslerian
spacetimes, J. Geom. Phys. \textbf{\ 61 } (2011) 1396--1400

\bibitem{vrflg} S. Vacaru, Finsler and Lagrange geometries in Einstein and
string gravity, Int. J. Geom. Meth. Mod. Phys. \textbf{5} (2008) 473-511


\bibitem{vspinor} S. Vacaru, Spinor structures and nonlinear connections in
vector bundles, generalized Lagrange and Finsler spaces, J. Math. Phys.
\textbf{37} (1996) 508-523

\bibitem{vhep} S. Vacaru, Spinors and field interactions in higher order
anisotropic spaces, JHEP, \textbf{09} (1998) 011, p. 1-49

\bibitem{vstavr} S. Vacaru and P. Stavrinos, Spinors and Space-Time
Anisotropy (Athens University Press, Athens, Greece, 2002), 301 pages,
gr-qc/0112028

\bibitem{cartf} E. Cartan, Les Espaces de Finsler (Paris, Herman, 1935)

\bibitem{vfricci} S. Vacaru, The entropy of Lagrange-Finsler spaces and
Ricci flows, Rep. Math. Phys. \textbf{63} (2009) 95-110

\bibitem{vsuperstr} S. Vacaru, Superstrings in higher order extensions of
Finsler superspaces, Nucl. Phys. B \textbf{434} (1997) 590 -656

\bibitem{vfncr} S. Vacaru, Spectral functionals, nonholonomic Dirac
operators, and noncommutative Ricci flows, J. Math. Phys. \textbf{50} (2009)
073503

\bibitem{vfbr} S. Vacaru, Finsler branes and quantum gravity phenomenology
with Lorentz symmetry violations, Class. Quant. Grav. \textbf{28} (2011)
215991


\bibitem{vcrit} S. Vacaru, Critical remarks on Finsler modifications of
gravity and cosmology by Zhe Chang and Xin Li, Phys. Lett. B \textbf{\ 690 }%
(2010) 224--228

\bibitem{vcluj} S. I. Vacaru, Minisuperspace twistor quantum cosmology.
Studia Universitatis Babes-Bolyai, Cluj-Napoca, Romania, XXXIV, \textbf{\ 2 }
(1990) 36--43

\bibitem{vostaf} S. Vacaru and S. Ostaf, Twistors and nearly autoparallel
maps, Rep. Math. Phys. \textbf{\ 37 }\ (1996) 309--324

\bibitem{vphd} S. Vacaru, Application of Nearly Authoparallel Maps and
Twistor--Gauge Methods in Gravity and Condensed States (Department of
Physics, University Alexandru Ioan Cuza, Ia\c{s}i, Romania, 1994) [in
Romanian]

\bibitem{vranceanu1} G. Vr\u{a}nceanu, Sur les espaces non holonomes. C. R.
Acad. Paris \textbf{103} (1926) 852--854

\bibitem{vranceanu2} G. Vr\u{a}nceanu, Sur quelques point de la th\'{e}ories
des espaces non holonomes. Bull. Fac. \c{S}t. Cern\u{a}u\c{t}i \textbf{5}
(1931) 177--205

\bibitem{vranceanu3} G. Vr\u{a}nceanu, Le\c{c}ons de Geometrie
Differentielle, Vol. II (Edition de l'Academie de la Republique Poopulaire
de Roumanie, 1957)

\bibitem{bejancu3} A. Bejancu and H. R. Farran, Foliations and Geometric
Structures (Springer, 2003)

\bibitem{kern} J. Kern, Lagrange Geometry, Archiv der Mathematik (Basel)
\textbf{25} (1974) 438--443

\bibitem{akbar} H. Akbar-Zadeh, Generalized Einstein manifolds, J. Geom.
Phys. \textbf{\ 17 } (1995) 342-380

\bibitem{ehresmann} C. Ehresmann, Les conexiones infinit\`{e}smales dans un
espace fibr\'{e} diff\'{e}rentiable. Coloque de Topologie, Bruxelles (1955)
29--55

\bibitem{kawaguchi} A. Kawaguchi, On the theory of non--linear connections,
I, II, Tensor, N. S. \textbf{ 2 } (1952) 123--142; \textbf{ 6 } (1956)
165--199

\bibitem{dunajski} S. Casey, M. Dunajski and P. Tod, Twistor geometry of a
pair of second order ODEs, arXiv: 1203.4158


\bibitem{vcosm} S. Vacaru, Principles of Einstein-Finsler Gravity and
Perspectives in Modern Cosmology, Int. J. Mod. Phys. D \textbf{ 21 } (2012) 1250072

\bibitem{berwald} L. Berwald, On Cartan and Finsler geometries, III. Two
dimensional Finsler spaces with rectilinear extrema, Ann. Math.\textbf{\ 42}
(1941) 84--122

\bibitem{chern} S. Chern, Local equivalence and Euclidean connections in
Finsler spaces, Sci. Rep. Nat. Tsing Hua Univ. Ser. A \textbf{5} (1948)
95--121; or Selected Papers, vol. II, 1994 (Springer, 1989)

\end{thebibliography}
\end{document}